\documentclass[11pt]{article}
\usepackage[letterpaper,margin=1in]{geometry}
\usepackage{amsmath,amssymb,amsthm}
\usepackage{graphicx,xcolor}
\usepackage[round,authoryear]{natbib}
\usepackage{hyperref}
\newtheorem{theorem}{Theorem}
\newtheorem{lemma}[theorem]{Lemma}
\newtheorem{proposition}[theorem]{Proposition}
\newtheorem{corollary}[theorem]{Corollary}

\newtheorem{remark}[theorem]{Remark}
\newenvironment{proofof}[1]{\begin{proof}\textbf{of {#1}}}{\end{proof}}
\newenvironment{keywords}{\par\smallskip\noindent\textbf{Keywords:}\ }{\par\medskip}
\newif\ifanonsubmission
\anonsubmissionfalse
\usepackage{newunicodechar}
\newunicodechar{’}{'}
\usepackage{times}
\setcitestyle{round}
\usepackage[T1]{fontenc}
\usepackage{mathtools}
\usepackage{placeins}
\usepackage{needspace}
\usepackage{longtable}
\usepackage{booktabs}
\usepackage{tikz}
\usepackage{comment}
\usetikzlibrary{decorations.pathreplacing,arrows.meta}
\usepackage{enumitem}
\usepackage[expansion=false]{microtype}
\hypersetup{colorlinks,linkcolor=blue!60!black,citecolor=blue!60!black,urlcolor=blue!60!black}

\newcommand{\rate}{I}

\newcommand{\batchreg}{R^{\mathrm{B}}}
\newcommand{\Xd}{\mathcal{X}_d}
\newcommand{\simplex}{\Delta_d}
\newcommand{\LSA}{\mathcal{M}_{d,L}}
\newcommand{\Exp}{\operatorname{Exp}}
\newcommand{\Dir}{\operatorname{Dirichlet}}
\newcommand{\Mult}{\operatorname{Multinomial}}
\newcommand{\profile}{\operatorname{profile}}
\newcommand{\E}{\mathbb{E}}
\newcommand{\Prob}{\mathbb{P}}
\newcommand{\cstar}{c^{\star}}
\newcommand{\Lmax}{L_{\max}}
\newcommand{\KL}{\operatorname{D}}
\renewcommand{\eqref}[1]{(\ref{#1})}

\title{A Layered Simplex Architecture for Large Alphabets}
\author{
  Meir Feder, Yaniv Fogel, and Itamar Shoshani\\
  School of Electrical and Computer Engineering\\
  Tel Aviv University, Tel Aviv, Israel\\
  \texttt{meir@tau.ac.il}, \texttt{yaniv.fogel8@gmail.com}, \texttt{ishoshani94@gmail.com}
  \and
  Ruediger Urbanke\\
  School of Computer and Communication Sciences\\
  EPFL, Lausanne, Switzerland\\
  \texttt{rudiger.urbanke@epfl.ch}
}

\date{}

\begin{document}
\maketitle

\begin{abstract}
A learning architecture induces a prior over the models it can express. Recent work suggests that a key property of successful architectures is a broad complexity spectrum: the prior probabilities of the expressible models span many orders of magnitude. Understanding how architectural structure gives rise to such priors, and how they in turn determine prediction, may help explain the success of modern machine learning architectures such as neural networks and transformers, with depth as the mechanism that produces the spectrum.
We investigate this connection in the analytically tractable setting of large-alphabet prediction. We introduce the \emph{layered simplex architecture} (LSA), in which independent uniformly distributed simplex vectors are multiplied coordinate-wise and normalized. At depth zero the prior is a point mass at the uniform distribution and at depth one the Bayesian predictor is Laplace's rule. 
 Each additional layer systematically moves prior mass toward more concentrated distributions.
For LSA, we derive exact expressions for the Bayesian marginal likelihood and regret using reusable kernels for each depth. We average over depths up to a cutoff that grows logarithmically with alphabet size. This gives an explicit predictor with rigorous performance guarantees. 
This architecture yields a broad complexity spectrum: for targets uniform on a specified sparse support, we prove that the mixture's prior complexity is asymptotic to the support-description cost at every fixed positive KL resolution. 
We also prove a regret lower bound for every label-invariant predictor: the expected description length of the discovered symbol set, minus the source's uncertainty about its identity given its size. For Zipf sources in a growing-alphabet regime, the depth mixture attains this bound to leading order. The rate of symbol discovery therefore determines the regret exponent.
Remarkably, although introduced as a simple model for studying the effect of depth, the resulting LSA predictor performs competitively with established large-alphabet methods, including Good--Turing, Ristad's law and an adaptive Dirichlet mixture, while substantially outperforming classical add-constant rules on concentrated sources.
\end{abstract}

\section{Introduction} 
\label{sec:intro}

A learning architecture, together with the randomness of its initialization, induces a \emph{prior over predictors} and the Bayesian regret on a target is governed by the prior mass near it~\citep{feder2025framework,Barron1987}. Recent work~\citep{feder2025framework} suggests that successful architectures have a \emph{broad complexity spectrum}: their prior probabilities span many orders of magnitude. Our long-term goal is to understand neural networks and transformers from this point of view, with depth as the mechanism that produces the spectrum. 

This paper follows the chain from architecture to prior to prediction in a simple setting of  estimating a distribution on a $d$-symbol alphabet from $N$ samples under log loss. 
The problem is hard when the alphabet is large relative to the data.
We propose and analyze the \emph{layered simplex architecture} (LSA) in which independent uniformly distributed simplex vectors are multiplied coordinate-wise and normalized. 
%Multiplying independent random layers spreads prior weights over many orders of magnitude.   
For this architecture, we can express the predictor and its regret exactly using integrals and count profiles. These expressions let us analyze how depth changes the prior and regret. 
Remarkably, although introduced as a simple model for studying the effect of depth, the resulting predictor performs competitively with established large-alphabet methods, including Good--Turing estimators~\citep{good1953population} designed for this regime, with strong guarantees~\citep{orlitsky2003always,mcallester2000leave,acharya2013tight,orlitsky2015competitive}, Ristad's law~\citep{ristad1995natural} and an adaptive Dirichlet mixture, while substantially outperforming classical add-constant rules, such as Laplace's rule and Krichevsky--Trofimov rule~\citep{krichevsky1981performance} on concentrated sources as these rules spread probability over the whole alphabet.

\paragraph{Contributions.}
\begin{enumerate}[leftmargin=*,itemsep=1pt,topsep=2pt]
\item \textbf{Exact performance expressions and their derivation} (Section~\ref{sec:regret}). Computing the Bayesian mixture with our proposed prior reduces to a one-dimensional integral. The resulting regret is evaluated by count-profile expressions, derived via the Gamma-function identity that decouples the normalizer, a layer recursion that reuses kernels across depths, a Mellin--Barnes contour for the kernels, and a generating function for the profile weights. Appendix~\ref{app:qlambda} bounds error propagation. Appendix~\ref{app:powers} extends the representation to layers raised to powers.
\item \textbf{Depth shapes the prior} (Section~\ref{sec:depth}). As $L$ grows the prior moves from dense to skewed and, at fixed $d$, converges as $L\to\infty$ to a uniform mixture over the simplex vertices; a large-deviation calculation motivates the scaling $L=c\ln d$. For targets uniform on $k=o(d)$ specified symbols, we prove that the depth mixture has prior complexity asymptotic to $\log_2\binom dk$ at every fixed positive KL resolution (Proposition~\ref{prop:complexity-spectrum}), establishing a broad complexity spectrum.
\item \textbf{Scaling laws from first principles and sharp leading-order regret} (Section~\ref{sec:scaling}). Empirical scaling laws describe power-law improvements with parameters and data~\citep{kaplan2020scaling}. Symbol discovery has long helped explain learning curves~\citep{hutter2021learning,michaud2023quantization,bahri2024explaining,debowski2025zipf}. One goal of the paper is to derive such laws for Bayesian estimators from first principles. For Zipf sources with fixed $\alpha>1$ and $d/N^{1/\alpha}\to\infty$, the expected number of distinct symbols grows as $N^{1/\alpha}$. Every predictor that treats symbol labels in the same way pays at least the expected number of bits needed to identify the observed symbols, minus the source's uncertainty about their identities given their number. For Zipf sources with $\alpha>1$ and $d=N^{\nu+o(1)}$, $\nu>1/\alpha$, the Bayesian depth mixture attains the sharp leading term $k_N\log_2(d/k_N)$, $k_N\asymp N^{1/\alpha}$ being the expected number of observed symbols: per-symbol regret of order $N^{1/\alpha-1}\log N$, with a leading coefficient no label-invariant predictor improves on.
\item \textbf{Empirical evidence} (Section~\ref{sec:experiments}). The depth average is competitive with Good--Turing, Ristad's law and an adaptive Dirichlet mixture, with large gains over add-constant rules on concentrated targets (Table~\ref{tab:bench}), and on the King James Bible (Table~\ref{tab:bible}).
\end{enumerate}
Regret and divergences are in bits unless marked in nats; $\ln$ denotes the natural logarithm.

\section{Architectures as priors}
\label{sec:setting}

Given a model class $\{p_\theta\}_{\theta\in\Theta}$ and a prior $w$, the Bayesian mixture is $Q_N(x^N)=\int w(d\theta)\,p_\theta(x^N)$. Fix a target $\theta_0$ and define its $\epsilon$-neighborhood by $\Theta_0^\epsilon=\{\theta:\KL(p_{\theta_0}\|p_\theta)\le\epsilon\}$. The mixture's expected extra log loss per symbol then satisfies the following target-dependent bound~\citep{Barron1987,feder2025framework}
\begin{equation}
\label{eq:priormass}
\tfrac1N\KL\bigl(p^N_{\theta_0}\,\big\|\,Q_N\bigr)\;\le\;\inf_{\epsilon\ge0}\Bigl[\epsilon-\tfrac1N\log w\bigl(\Theta_0^\epsilon\bigr)\Bigr].
\end{equation}
Appendix~\ref{app:regret} proves the bound. We call $-\log_2 w(\Theta_0^\epsilon)$ the \emph{complexity of the target under the prior}. For regular $k$-parameter iid models with a continuous positive prior density near an interior target, cumulative regret has the familiar $\frac{k}{2}\log N$ leading term~\citep{clarke1990information}. Here the architecture determines the prior and its target-dependent complexity~\citep{feder2025framework}.
An architecture-induced prior has a \emph{broad complexity spectrum} when the neighborhood masses $w(\Theta_0^\epsilon)$ span many orders of magnitude 
as $\theta_0$ moves across the model class; equivalently, the corresponding negative-log masses vary over a broad range.

\begin{comment}
\paragraph{The LSA prior.}
Let $\Xd=\{1,\dots,d\}$ and $\simplex$ be the probability simplex. For $L\ge1$, denote the prior in~\eqref{eq:lsa} by $\LSA$. Each uniform simplex layer is a normalized vector of independent exponentials; its normalization cancels in~\eqref{eq:lsa}, giving
\begin{equation}
\label{eq:exprep}
Y_i=\prod_{\ell=1}^{L}E_{\ell i},\qquad \theta_i=\frac{Y_i}{\sum_{j}Y_j},\qquad E_{\ell i}\stackrel{\mathrm{iid}}{\sim}\Exp(1).
\end{equation}
\end{comment}

\paragraph{The layered simplex architecture.}
Let $\Xd=\{1,\dots,d\}$, and let $\simplex$ denote the probability simplex. For $L\ge1$, draw $U^{(1)},\dots,U^{(L)}\stackrel{\mathrm{iid}}{\sim}\Dir(1,\dots,1)$, multiply coordinate-wise, and normalize (Figure~\ref{fig:lsa}):
\begin{equation}
\label{eq:lsa}
\theta_i \;=\; \frac{\prod_{\ell=1}^{L}U^{(\ell)}_i}{\sum_{j=1}^{d}\prod_{\ell=1}^{L}U^{(\ell)}_j},\qquad i=1,\dots,d.
\end{equation}
We denote the resulting distribution of $\theta$ by $\LSA$. In Hinton's terminology~\citep{Hinton1999ProductsOE}, this construction is a random product of experts, with each $U^{(\ell)}$ acting as one expert. Depth---the number of layers---changes the prior distribution over a fixed model class. At every positive finite depth, the prior assigns positive probability to every open region of the simplex. The predictor $Q_N^{(L)}$ is the exact Bayesian mixture of multinomial sources under this prior. At $L=0$, the empty product gives $1$ for each component, so the prior is a point mass at $\theta_i=1/d$ and $Q_N^{(0)}(x^N)=d^{-N}$.

A uniform Dirichlet vector can be written as a normalized vector of independent \emph{exponential} variables. The layer normalizations cancel in the final normalization. Therefore, we can write the same prior as
\begin{equation}
\label{eq:exprep}
Y_i=\prod_{\ell=1}^{L}E_{\ell i},\qquad \theta_i=\frac{Y_i}{\sum_{j}Y_j},\qquad E_{\ell i}\stackrel{\mathrm{iid}}{\sim}\Exp(1).
\end{equation}
This representation makes the architecture explicit. The final step is Bayesian averaging over depths; Section~\ref{sec:depth} quantifies the resulting range of target complexities through prior neighborhood masses.
%It forms $d$ independent products, each containing $L$ nonnegative factors, and then normalizes them to sum to one.

\begin{figure}[!t]
\centering
\resizebox{0.75\columnwidth}{!}{\begin{tikzpicture}[x=4.0mm,y=1mm]
  % dotted guide through the winning coordinate (drawn in two segments)
  \draw[black!35,densely dotted,line width=0.5pt] (8.31,8.5) -- (8.31,53);
  \draw[black!35,densely dotted,line width=0.5pt] (8.31,-18) -- (8.31,-3.5);
  % --- layer rows ---
  \foreach \i/\h in {1/0.87,2/1.27,3/0.70,4/1.11,5/0.26,6/4.18,7/0.01,8/3.47,9/0.71,10/0.37,11/0.67,12/0.39,13/1.11,14/1.33,15/2.33,16/0.27,17/3.89,18/0.91,19/0.43,20/1.09,21/0.09,22/0.07,23/1.51,24/0.96} \fill[black!45] (\i,44) rectangle ++(0.62,\h);
  \draw[black!60,line width=0.4pt] (0.9,44) -- (25.0,44);
  \node[left,font=\small] at (0.7,46.4) {$U^{(1)}$};
  \foreach \i/\h in {1/2.45,2/0.53,3/0.58,4/1.32,5/0.59,6/0.02,7/0.33,8/1.13,9/0.71,10/1.56,11/0.01,12/1.01,13/0.28,14/0.40,15/3.43,16/1.57,17/4.23,18/1.95,19/2.50,20/0.16,21/0.34,22/0.27,23/1.28,24/1.36} \fill[black!45] (\i,30) rectangle ++(0.62,\h);
  \draw[black!60,line width=0.4pt] (0.9,30) -- (25.0,30);
  \node[left,font=\small] at (0.7,32.4) {$U^{(2)}$};
  \node[font=\small] at (13,24.5) {$\vdots$};
  \foreach \i/\h in {1/0.39,2/1.87,3/0.74,4/0.83,5/0.12,6/0.27,7/2.25,8/0.80,9/0.03,10/1.46,11/0.54,12/3.25,13/0.93,14/1.07,15/0.48,16/0.02,17/5.92,18/0.52,19/2.28,20/0.16,21/3.02,22/0.02,23/0.94,24/0.10} \fill[black!45] (\i,10) rectangle ++(0.62,\h);
  \draw[black!60,line width=0.4pt] (0.9,10) -- (25.0,10);
  \node[left,font=\small] at (0.7,12.4) {$U^{(L)}$};
  % --- braces ---
  \draw[decorate,decoration={brace,amplitude=3pt},black!70]
    (1.0,52.5) -- (24.62,52.5) node[midway,above=3pt,font=\small] {$d$ coordinates, one per symbol};
  \draw[decorate,decoration={brace,amplitude=3pt},black!70]
    (25.6,50) -- (25.6,10) node[midway,right=4pt,align=left,font=\small]
    {$L$ independent\\ uniform draws\\ from $\Delta_d$};
  % --- arrow + operation ---
  \draw[-{Stealth[length=2.4mm]},black!80] (12.81,7) -- (12.81,-2.5);
  \node[right=8pt,align=left,font=\small] at (13.0,2.2)
    {multiply coordinatewise, renormalize};
  % --- theta row ---
  \foreach \i/\h in {1/0.11,2/0.52,3/0.35,4/0.09,5/0.00,6/0.03,7/0.02,8/9.18,9/0.01,10/1.19,11/0.00,12/1.28,13/0.05,14/1.53,15/3.65,16/0.01,17/2.78,18/1.08,19/2.82,20/0.03,21/0.06,22/0.00,23/3.12,24/0.09} \fill[black!80] (\i,-16) rectangle ++(0.62,\h);
  \draw[black!60,line width=0.4pt] (0.9,-16) -- (25.0,-16);
  \node[left,font=\small] at (0.7,-13.6) {$\theta$};
  \node[right,black!60,font=\small,align=left] at (0.9,-6.2)
    {sparse, heavy-tailed:\\ a few coordinates \\ carry the mass};
\end{tikzpicture}}
\caption{The layered simplex architecture. Each layer is an independent uniform point of the simplex. The panels show layers $1$, $2$, and $L$ (dots indicate omitted layers); the coordinate-wise product uses all $L$ layers and is normalized to give $\theta\sim\LSA$. The dotted line marks the largest coordinate of $\theta$.}
\label{fig:lsa}
\end{figure}

\paragraph{Extension: Powered layers.}
A natural extension of~\eqref{eq:exprep} uses $Y_i=\prod_{\ell=1}^{L}E_{\ell i}^{w_{\ell i}}$ with positive powers $w_{\ell i}$, keeping the same independent exponential variables and normalization. The choice $w\equiv1$ recovers~\eqref{eq:exprep}; a power above one sharpens the factor, below one flattens it. The powers may be \emph{random}, $w_{\ell i}\stackrel{\mathrm{iid}}{\sim}\nu$ with $\nu$ part of the prior and integrated out, which keeps the prior label-invariant, or \emph{fixed}, $w_{\ell i}=\beta_\ell$ shared within each layer. 
%for which $\E Y^z=\prod_\ell\Gamma(1+\beta_\ell z)$ and a common power $w$ is the escort transform $\theta_i(w)\propto\theta_i(1)^w$ of a unit-power draw. 
One layer at an integer power $k$ is in distribution the product of $k$ Dirichlet layers with parameters $1/k,\dots,1$ (Theorem~\ref{thm:integer-powers}, Appendix~\ref{app:powers}). As in LSA, a final step may average over depths and powers.
Sections~\ref{sec:depth}--\ref{sec:experiments} focus on LSA with unit powers. Appendix~\ref{app:powers} develops the extension to other powers and reports its experimental results.
%Sections~\ref{sec:depth}--\ref{sec:scaling} concern unit powers.

\section{Exact Bayesian predictor and its regret via count profiles}
\label{sec:regret}

The LSA predictor assigns probability $Q_N(x^N)=\E_{\theta\sim\LSA}\bigl[\prod_{t=1}^N\theta_{x_t}\bigr]$ to the sequence $x^N$. Let $m_i$ be the number of times symbol $i$ appears. Then $Q_N(x^N)=q_m:=\E_{\theta\sim\LSA}\bigl[\prod_i\theta_i^{m_i}\bigr]$ depends on $x^N$ only through the count vector $m$. Because the prior is unchanged by relabeling the symbols, $Q_N(x^N)$ depends on $m$ only through its \emph{profile} $\lambda=\profile(m)=\{c_0,c_1,\dots,c_N\}$~\citep{orlitsky2004modeling}. Here, $c_r$ is the number of symbols that appear $r$ times, so $0\le c_r\le d$, $\sum_rc_r=d$ and $\sum_rrc_r=N$; the unseen class has $c_0=d-s$ with $s=\sum_{r\ge1}c_r$ distinct symbols observed. Using any fixed grouping of the coordinates into groups of sizes $c_0,\dots,c_N$, write $\theta_{r,j}$ for coordinate $j$ in group $r$. Then $Q_N(x^N)=q_m=q_\lambda:=\E_{\theta\sim\LSA}\bigl[\prod_{r=0}^N\prod_{j=1}^{c_r}\theta_{r,j}^{\,r}\bigr]$; the $c_0$ unused coordinates contribute factors of one. Let $\Lambda_N$ denote the set of possible profiles for $N$ observations.

\begin{theorem}[profile decomposition of the regret]
\label{thm:regret}
Fix $d\ge2$ and $L\ge0$. The LSA sequence probabilities are positive, invariant under observation order and alphabet relabeling, and consistent under marginalization: $\sum_{a=1}^d Q_{N+1}(x^N a)=Q_N(x^N)$. The empty-sample probability is $q_0=1$. For an iid target $p\in\simplex$ and $N\ge1$, the per-symbol online regret $R_N(Q,p)=\frac1N\KL(p^N\|Q_N)$ is
\begin{equation}
\label{eq:regret}
R_N(Q,p)=-H(p)-\frac1N\sum_{\lambda\in\Lambda_N}A_\lambda(p)\log_2 q_\lambda,
\end{equation}
where $A_\lambda(p)=\Prob_{M\sim\Mult(N,p)}\bigl(\profile(M)=\lambda\bigr)$.
For $N\ge0$, the batch regret $\batchreg_N(Q,p)=\E\bigl[\KL\bigl(p\,\|\,\hat q_{M}\bigr)\bigr]$ of the predictive distribution $\hat q_m(i)=q_{m+e_i}/q_m$ is
\begin{equation}
\label{eq:batch}
\batchreg_N(Q,p)=-H(p)-\E\Bigl[\sum_{\substack{r\ge0:\,c_r(M)>0}}S_r(M)\log_2\frac{q_{\lambda\oplus r}}{q_\lambda}\Bigr],
\end{equation}
where $S_r(m)=\sum_{i:m_i=r}p_i$, $\lambda\oplus r$ moves one symbol of $\lambda$ from count $r$ to $r+1$, and $NR_N=\sum_{n=0}^{N-1}\batchreg_n$ for $N\ge1$. The sum includes the unseen class $r=0$ whenever it is nonempty.
\end{theorem}
Appendix~\ref{app:regret} gives the proof. In the online identity~\eqref{eq:regret}, the target enters through $H(p)$ and $A_\lambda(p)$; the next-symbol identity~\eqref{eq:batch} also uses the count-class masses $S_r(M)$. The architecture enters through the evidence values $q_\lambda$. For a fixed count vector, an LSA mixture assigns the same predictive probability to symbols with the same count. That probability may depend on the full count profile. Every LSA mixture is therefore a \emph{natural estimator}~\citep{orlitsky2015competitive}. Its regret is at least that of the natural oracle in Section~\ref{sec:experiments}; Appendix~\ref{app:regret} gives the exact decomposition. 

\begin{comment}
Appendix~\ref{app:regret} gives the proof. In the online identity~\eqref{eq:regret}, the target enters through $H(p)$ and $A_\lambda(p)$; the next-symbol identity~\eqref{eq:batch} also uses the count-class masses $S_r(M)$. The architecture enters through the evidence values $q_\lambda$. For a fixed count vector, symbols with the same count receive the same predictive probability. This probability may still depend on the full count profile. An estimator with this property is called a \emph{natural estimator}~\citep{orlitsky2015competitive}. Its regret is at least that of the best natural estimator, called the natural oracle, discussed in Section~\ref{sec:experiments}. Both $H(p)$ and $A_\lambda(p)$ are unchanged by relabeling the symbols. Therefore, the regret depends on $p$ only through its sorted vector of probabilities, which we call its \emph{shape}.

\paragraph{Computing the two factors.} The regret formula requires two profile-dependent quantities: $q_\lambda$, the mixture probability which depends on the  prior, and $A_\lambda(p)$, which comes from the target. 
\end{comment}

\paragraph{The mixture \texorpdfstring{$q_\lambda$}{q} as a one-dimensional integral.}
Direct evaluation of $q_\lambda$ requires integrating over the layered prior on the $(d-1)$-dimensional simplex. In~\eqref{eq:exprep}, use the same grouping as above: $\theta_{r,j}=Y_{r,j}/S$, where $S$ is the sum of all $d$ unnormalized weights. For $N\ge1$, substituting $u=tS$ gives $\Gamma(N)=\int_0^\infty u^{N-1}e^{-u}\,du=S^N\int_0^\infty t^{N-1}e^{-tS}\,dt$. Hence $S^{-N}=\Gamma(N)^{-1}\int_0^\infty t^{N-1}e^{-tS}\,dt$. Since $S>0$ almost surely and the integrand is nonnegative, Tonelli's theorem lets us interchange expectation and integration. Independence of the $Y_{r,j}$ then gives
\begin{equation}
\label{eq:gammatrick}
\begin{aligned}q_\lambda&=\E\Bigl[S^{-N}\prod_{r=0}^N\prod_{j=1}^{c_r}Y_{r,j}^{\,r}\Bigr]\\
&=\frac{1}{\Gamma(N)}\int_0^\infty t^{N-1}\prod_{r=0}^N\prod_{j=1}^{c_r}\phi_r(t)\,dt
=\frac{1}{\Gamma(N)}\int_0^\infty t^{N-1}\prod_{r=0}^N\phi_r(t)^{c_r}\,dt,\end{aligned}
\end{equation}
where $\phi_r(t)=\E[Y^re^{-tY}]$ with $Y=\prod_{\ell=1}^LE_\ell$ the unnormalized weight of one symbol. After factorization, the $c_r$ symbols in group $r$ each contribute $\phi_r(t)$, so their product is $\phi_r(t)^{c_r}$. The unseen group $r=0$ also contributes, because its weights enter $S$. The architecture enters through these one-symbol kernels. The kernels obey a recursion over the layers. Writing $Y=E\cdot V$ with $E\sim\Exp(1)$ the last layer and $V$ the product of the other $L-1$, conditioning on $E=x$ gives $\E[Y^re^{-tY}\mid E=x]=x^r\phi_r^{(L-1)}(tx)$, so
\begin{equation}
\label{eq:phi-recursion-main}
\phi^{(\ell)}_r(t)=\int_0^\infty e^{-x}x^r\,\phi^{(\ell-1)}_r(tx)\,dx,\qquad \phi^{(1)}_r(t)=\frac{\Gamma(r+1)}{(1+t)^{r+1}} .
\end{equation}
We use the Mellin--Barnes representation to compute kernel tables once for each depth and count. At each depth, computing the evidence requires a one-dimensional integral. The observed profile and each profile obtained by adding one observation use the same kernels and integration grid.
Depths zero and one use closed forms. Appendices~\ref{app:qlambda}
and~\ref{app:validation} give the numerical methods, error bounds and checks.

\paragraph{The profile weights \texorpdfstring{$A_\lambda(p)$}{A}.}
Under $M\sim\Mult(N,p)$, $\Prob(M=m)=N!\prod_ip_i^{m_i}/m_i!$. Let $\mathcal R_\lambda=\{r\ge1:c_r>0\}$. Each symbol is either unused, contributing $1$, or used $r$ times, contributing $z_rp_i^r/r!$, so summing over the count vectors with profile $\lambda$ is coefficient extraction in a generating function:
\begin{equation}
\label{eq:alambda-main}
A_\lambda(p)=\Gamma(N+1)\,[z^c]\prod_{i=1}^d\Bigl(1+\sum_{r\in\mathcal R_\lambda}z_r\frac{p_i^r}{\Gamma(r+1)}\Bigr),
\end{equation}
where $[z^c]$ denotes the coefficient of $\prod_{r\in\mathcal R_\lambda}z_r^{c_r}$ and the unseen class is the constant term. By symmetry, the coefficient is a polynomial in the power sums $P_j=\sum_i p_i^j$; it is evaluated recursively and is explicit for the uniform target (Appendix~\ref{app:alambda}). Computing the exact sum~\eqref{eq:regret} is impractical at the tested sizes because it requires listing all count profiles, or partitions of $N$ into at most $d$ parts. We therefore draw $M\sim\Mult(N,p)$ using fixed random seeds and average the computed values of $\log_2q_M$. This estimate is unbiased when the log-probabilities are exact. The error bars show sampling variation across profiles.

\paragraph{Exact predictor and regret for powered layers.}
For the powered model, the unnormalized coordinates remain independent, so~\eqref{eq:gammatrick} holds verbatim and only the kernels change: for a fixed layer power the recursion~\eqref{eq:phi-recursion-main} becomes $\phi^{(\ell)}_r(t)=\int_0^\infty e^{-x}x^{\beta_\ell r}\phi^{(\ell-1)}_r(tx^{\beta_\ell})\,dx$, and for random powers the layer step is the same integral averaged over $w\sim\nu$. Appendix~\ref{app:powers} gives the details and the performance 
%compares power mixtures with depth mixtures 
on the synthetic benchmark.
\begin{comment}
Fixed powers $\beta_\ell>0$, with $Y_i=\prod_\ell E_{\ell i}^{\beta_\ell}$, preserve tractable evidence. Appendix~\ref{app:powers} derives the extension and its Gamma representations, and compares single-layer fixed-power mixtures with depth mixtures on the synthetic benchmark.
\end{comment}

\section{The role of the depth \texorpdfstring{$L$}{L}}
\label{sec:depth}

Depth changes how the architecture distributes prior mass over the simplex. Averaging over depths produces a single prior with a broad range of target complexities, measured by the negative logarithm of prior neighborhood mass. We establish this property and examine its consequences for prediction.

\paragraph{Depth skews the prior.}

At $L=0$, the prior is a point mass at the uniform distribution.

\begin{proposition}[shallow end]
\label{prop:laplace}
For $L=1$, $q_m=\Gamma(d)\prod_i m_i!/\Gamma(d+N)$. The resulting sequential predictor is Laplace's add-one rule $Q(x_{n+1}=i\mid x^n)=(m_i+1)/(n+d)$.
\end{proposition}

Recall \eqref{eq:exprep}. For fixed, small $L$,
most coordinates of $\theta$ are non-negligible. The variable $\ln Y_i$ is a sum of $L$ iid terms, each with mean $-\gamma$ and variance $\pi^2/6$. For large $L$, it is therefore approximately distributed as $\mathcal{N}(-\gamma L,L\pi^2/6)$. The largest $Y_i$ values increasingly dominate the normalization, so the probabilities become more uneven as $L$ grows. 
For fixed $d$, as $L\to\infty$, the prior converges to a uniform mixture over the simplex vertices. Equivalently, the predictor assigns vanishing probability to any sequence containing two distinct symbols (Appendix~\ref{app:fixedd}). 

\paragraph{The logarithmic depth scale.}
The large-deviation calculation in Appendix~\ref{app:fixedd} singles out $\cstar=1/(1-\gamma)$, where $\gamma\approx0.57721$ is the Euler--Mascheroni constant, and motivates comparing depths on the scale $L=c\ln d$. Related phase-transition analysis is discussed by \citet{jeong2026phase}. The proof of Theorem~\ref{thm:optimal} is self-contained.

\paragraph{Using a mixture over depths.}
Different targets may favor different depths. We therefore also study the depth-averaged predictor
\begin{equation}
\label{eq:avg}
Q^{\mathrm{avg}}_N(x^N)\;=\;\frac{1}{\Lmax+1}\sum_{L=0}^{\Lmax}Q^{(L)}_N(x^N).
\end{equation}
This gives a single Bayesian predictor whose set of depths is fixed before the data are observed. More generally, mixing over depths gives the following guarantee.
\begin{proposition}[depth averaging]
\label{prop:avg}
Let $Q_N=\sum_L\pi_LQ^{(L)}_N$ for a depth distribution $\pi$ fixed in advance. For every sequence, every $L$ with $\pi_L>0$, and every target,
\begin{align*}
-\log_2 Q_N(x^N)&\le -\log_2 Q^{(L)}_N(x^N)+\log_2(1/\pi_L),\\
R_N(Q,p)&\le R_N(Q^{(L)},p)+\tfrac{1}{N}\log_2(1/\pi_L).
\end{align*}
For the uniform prior in \eqref{eq:avg}, the added cumulative log loss is at most $\log_2(\Lmax+1)$. The prior $\pi_L=1/[L(L+1)]$ uses all positive depths and has unbounded depth support. Its penalty at depth $L$ is $\log_2[L(L+1)]$.
\end{proposition}

The penalty is paid once per sequence; the posterior weights $w_L(x^N)\propto\pi_LQ_N^{(L)}(x^N)$ adapt to the data.

\paragraph{Complexity spectrum.}
The following proposition shows that the depth mixture assigns widely different neighborhood probabilities to different targets. We establish this for targets uniform on specified sparse supports.
\begin{proposition}[complexity spectrum of the depth mixture]
\label{prop:complexity-spectrum}
Fix $\epsilon>0$, measured in bits. Let $W_d=(\Lmax+1)^{-1}\sum_{L=0}^{\Lmax}\mathcal M_{d,L}$ be the prior inducing~\eqref{eq:avg}, with $\Lmax=\lceil\cstar\ln(2d)\rceil+1$, and define $C_\epsilon(p)=-\log_2 W_d\{\theta:\KL(p\|\theta)\le\epsilon\}$. For every $1\le k=o(d)$ and every $k$-element subset $S$, let $u_S$ be uniform on $S$. Then, as $d\to\infty$,
\begin{equation}
\label{eq:complexity-spectrum}
C_\epsilon(u_S)\sim\log_2\binom dk\sim k\log_2(d/k).
\end{equation}
\end{proposition}
As the support size varies, the neighborhood probabilities of these targets span many orders of magnitude under the same depth-mixture prior. This establishes a broad complexity range in the sense defined in Section~\ref{sec:setting}. The comparison uses the same fixed $\epsilon>0$ for every target. The leading term is independent of this tolerance because the mixture includes depths of order $\ln(d/k)$ that put enough mass near sparse targets; $\epsilon$ affects only lower-order terms. Appendix~\ref{app:optimal} gives the proof.

\begin{figure*}[!htbp]
\centering
\includegraphics[width=0.95\textwidth]{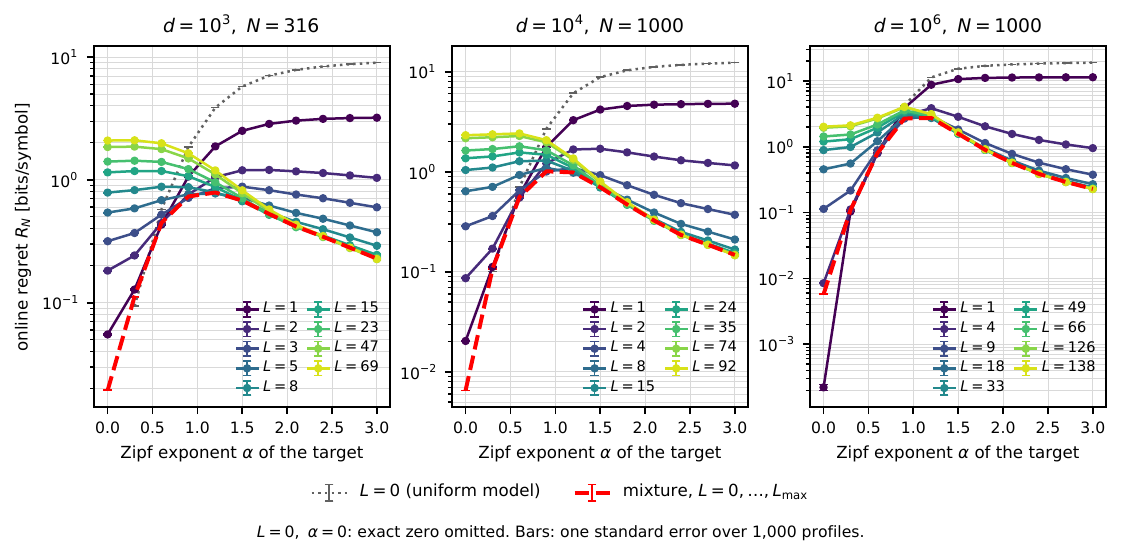}
\caption{Average cumulative regret per symbol, $R_N$ in bits (logarithmic vertical axis), versus Zipf exponent $\alpha$. Panels use $(d,N)=(10^3,316),(10^4,1000),(10^6,1000)$. Colored curves show representative depths; the red dashed curve averages sequence probabilities with equal prior weights over all integer depths $0$--$\Lmax$, with $\Lmax=69,92,138$, respectively. Means and one-standard-error bars use 1,000 sampled profiles per point, shared across depths. At $\alpha=0$, the $L=0$ regret is exactly zero and is omitted from the log axis. Sampling and numerical protocols are in Appendices~\ref{app:experiments} and~\ref{app:validation}.}
\label{fig:spectrum}
\end{figure*}

\paragraph{Depth and prediction.}
Figure~\ref{fig:spectrum} illustrates why we average over depths. Small depths perform best when the target is close to uniform, while larger depths generally help when probability is concentrated on fewer symbols. For example, at $d=10^6$ and $\alpha=3$, the best tested depth gives a regret of $0.23$ bits per symbol, compared with $11.36$ for $L=1$---about a fifty-fold reduction. The depth mixture adapts by giving more weight to depths that better explain the observations. Its regret exceeds that of the best included depth by at most $\log_2(\Lmax+1)/N$ bits per symbol.

\FloatBarrier
\section{Scaling laws from symbol discovery}
\label{sec:scaling}

The support-description cost in Proposition~\ref{prop:complexity-spectrum} reappears in prediction through the set of symbols discovered by the sample. This section has three steps. First, we count how many symbols the sample reveals. Second, we derive a lower bound that every label-invariant predictor must satisfy. Third, we show that the depth mixture matches the leading term of this bound.  

\paragraph{How many symbols are discovered.} Let $S_N$ be the set of distinct symbols observed in $N$ independent samples from $p$, $K_N=|S_N|$, and $k_N=\E[K_N]=\sum_{i\le d}[1-(1-p_i)^N]$. For a Zipf target $p=p_{d,\alpha}$, $p_{d,\alpha}(i)=i^{-\alpha}/H_{d,\alpha}$ with $H_{d,\alpha}=\sum_{j=1}^dj^{-\alpha}$, and $\alpha>1$, set $a=N/H_{d,\alpha}$ and $a_N=a^{1/\alpha}$. Proposition~\ref{prop:discovery} gives $|k_N-\Gamma(1-1/\alpha)a_N|\le1+\epsilon_P+ad^{1-\alpha}/(\alpha-1)$ with $\epsilon_P\le[2e(1-\max_ip_i)]^{-1}$, so as $N\to\infty$ with $d/N^{1/\alpha}\to\infty$, $k_N\sim\Gamma(1-1/\alpha)(N/\zeta(\alpha))^{1/\alpha}$~\citep{karlin1967central,gnedin2007notes}: when the discovered set is small relative to the alphabet, the expected number of discovered symbols is of order $N^{1/\alpha}$.

\paragraph{What every predictor must pay.}
Under a label-invariant predictor all $\binom{d}{K_N}$ sets of $K_N$ observed symbols have the same probability, so identifying the observed set costs $\log_2\binom{d}{K_N}$ bits. But the source also makes this set random, with uncertainty $H_P(S_N\mid K_N)$ given $K_N$. The theorem below subtracts this source uncertainty from the naming cost; for a uniform source the two cancel, as they should, since the uniform predictor has zero regret.

\begin{theorem}[discovered-set decomposition]
\label{thm:support}
Let $P_N=p^N$, and let $Q_N$ be any label-invariant distribution on $\Xd^N$ that assigns positive probability to every sequence. Then
\begin{align}
\label{eq:support}
\begin{aligned}\KL(P_N\|Q_N)={}&\KL(P_{K_N}\|Q_{K_N})+\E\bigl[\log_2\tbinom{d}{K_N}\bigr]\\&-H_P(S_N\mid K_N)+\E_{S_N}\bigl[\KL\bigl(P_{X^N\mid S_N}\,\big\|\,Q_{X^N\mid S_N}\bigr)\bigr].\end{aligned}
\end{align}
Both divergence terms are nonnegative. Therefore, $\KL(P_N\|Q_N)\ge\E\bigl[\log_2\binom{d}{K_N}\bigr]-H_P(S_N\mid K_N)$.
\end{theorem}
Appendix~\ref{app:support} proves the result by applying the KL chain rule first to $S_N$ and then to $K_N$. For Zipf targets, the subtracted entropy term is only $O(k_N)$:
\begin{corollary}[Zipf lower bound]
\label{cor:zipf}
For $p=p_{d,\alpha}$, $\alpha>1$, $a_N\ge1$, and every label-invariant $Q_N$,
\begin{equation}
\label{eq:zipflb}
\KL(p^N\|Q_N)\;\ge\;\frac{k_N\ln(d/k_N)-C_\alpha a_N-2}{\ln 2}.
\end{equation}
Here $C_\alpha=\ln2+\frac{1}{\alpha-1}+\frac{\alpha}{(\alpha-1)^2}$. If $a_N\to\infty$ and $a_N/d\to0$, then $\KL(p^N\|Q_N)\ge k_N\log_2(d/k_N)-O_\alpha(k_N)$.
\end{corollary}
If $d/N^{1/\alpha}\to\infty$, every label-invariant predictor therefore satisfies $R_N\gtrsim N^{1/\alpha-1}\log_2(d/N^{1/\alpha})$; the logarithmic factor is $\log(d/k_N)$, which simplifies to $\log d$ only when $\log N=o(\log d)$.

\paragraph{What the depth mixture achieves.}
For Zipf targets with $d/N^{1/\alpha}\to\infty$, Corollary~\ref{cor:zipf} gives the lower bound $(1-o(1))k_N\log_2(d/k_N)$ for every label-invariant predictor. The next theorem shows that the depth mixture matches this cost when the alphabet grows polynomially faster than the discovered set.
\begin{theorem}[first-order optimality of depth averaging]
\label{thm:optimal}
Fix $\alpha>1$. Let $N\to\infty$, and assume that $d=N^{\nu+o(1)}$ for some $\nu>1/\alpha$. Let $p=p_{d,\alpha}$. Let $Q_N$ be either the uniform depth average~\eqref{eq:avg}, with $\Lmax=\lceil\cstar\ln(2d)\rceil+1$, or the mixture with $\pi_L=1/[L(L+1)]$. Then
\begin{equation}
\label{eq:optimal}
\KL(p^N\|Q_N)=(1+o(1))\,k_N\log_2(d/k_N),\qquad
R_N(p)\sim\frac{\Gamma(1-1/\alpha)}{\zeta(\alpha)^{1/\alpha}}\,N^{1/\alpha-1}\log_2\frac{d}{N^{1/\alpha}}.
\end{equation}
\end{theorem}
The expected number of discovered symbols grows as $N^{1/\alpha}$. Regret per symbol therefore decays as $N^{-(1-1/\alpha)}$, multiplied by $\log_2(d/k_N)$, the leading description cost per discovered symbol. No label-invariant predictor can improve this leading term. By symmetrization, it is also the smallest possible leading term for worst-case regret over all relabelings of the Zipf target (Corollary~\ref{cor:permutation-minimax}).

\emph{Proof idea.} For the bound, we choose depth $\cstar\ln(d/k)+O(1)$. We then consider an event under the prior in which the $k$ observed weights lie near $2dm_i/k$ and the total unseen weight is bounded. Exponential tilting bounds the expected excess code length above empirical entropy by $k_N\ln(d/k_N)+O(k_N\ln\ln d)$ nats. Mixing adds $O(\ln\ln d)$; entropy concavity and Corollary~\ref{cor:zipf} finish the argument (Appendix~\ref{app:optimal}). This comparator depth decreases with $k$, consistently with the observed downward shift for Zipf $\alpha=1.5$ in Section~\ref{sec:experiments}. It serves to bound cumulative regret; finite-sample evidence weights adapt to the full count profile.

Appendix~\ref{app:optimal} also allows larger depth cutoffs with negligible mixture penalty (Remark~\ref{rem:larger-cutoffs}); Appendix~\ref{app:experiments} examines finite-sample scaling.

% Updated empirical section; accepted tables retain their numerical entries and bolding.
\section{Competitive comparison}
\label{sec:experiments}

\paragraph{Setup.}
We measure next-symbol prediction by the expected KL loss
$\E\bigl[\KL(p\|\widehat q_M)\bigr]$ in bits after $n$ observations. This is the batch regret defined in~\eqref{eq:batch}. The eleven targets include
uniform, step, Zipf, geometric and random Dirichlet distributions, extending
\citet{orlitsky2015competitive}. Table~\ref{tab:bench} uses $n=1000$ and
$d=10^4$, with 1,000 shared trials per target. We compare fixed depth $L=22$ with a mixture over $L=0,\ldots,80$ whose initially equal weights adapt to the data. All comparison methods, including a symmetric Dirichlet concentration mixture, use the same samples. The natural
oracle knows the target but must assign equal probability to symbols with
identical counts. Appendix~\ref{app:experiments} gives the common sampling
protocol, uncertainty conventions and estimator definitions.

\begin{table*}[!htbp]
\centering
\caption{Mean next-symbol KL loss in bits at $n=1000$, $d=10^4$, over 1,000
shared trials per target. LSA uses fixed depth $22$ or equal prior weights
on depths $0,\ldots,80$, including the uniform model. GT is the
Good--Turing/empirical hybrid, Dir-$\tau$ the Dirichlet concentration mixture,
and AD robust absolute discounting. Bold marks the smallest non-oracle mean at four
decimal places, including rounding ties; the target-informed natural oracle
is excluded. No trial is omitted. Numerical predictions are normalized
before scoring. Protocol and uncertainty: Appendix~\ref{app:experiments}.}
\label{tab:bench}
\footnotesize\setlength{\tabcolsep}{3.2pt}
\begin{tabular}{lrrrr|rr|rr|r}
\toprule
& \multicolumn{4}{c|}{classical} & \multicolumn{2}{c|}{adaptive controls} & \multicolumn{2}{c|}{LSA} & \\
target & add-1 & KT & Ristad & GT & Dir-$\tau$ & AD & $L=22$ & avg $0$--$80$ & oracle\\
\midrule
uniform & 0.0395 & 0.1087 & 0.0020 & 0.0625 & 0.0011 & 0.0552 & 1.4364 & \textbf{0.0001} & 0.0000 \\
step & 0.2048 & 0.2608 & 0.1865 & 0.2268 & \textbf{0.1853} & 0.2477 & 1.5275 & 0.1896 & 0.1846 \\
Zipf $\alpha=1$ & 1.7370 & 1.3533 & 1.0127 & 0.6577 & 0.8522 & \textbf{0.6453} & 0.9525 & 0.6688 & 0.6280 \\
Zipf $\alpha=1.5$ & 3.0558 & 2.3047 & 0.5770 & 0.3701 & 0.4271 & \textbf{0.3586} & 0.3707 & 0.3615 & 0.3454 \\
Zipf $\alpha=2$ & 3.3625 & 2.5263 & 0.2359 & 0.1650 & 0.1720 & 0.1587 & \textbf{0.1582} & 0.1586 & 0.1488 \\
Zipf $\alpha=3$ & 3.4400 & 2.5758 & 0.0662 & 0.0499 & 0.0489 & 0.0477 & 0.0497 & \textbf{0.0469} & 0.0417 \\
Zipf $\alpha=4$ & 3.4506 & 2.5813 & 0.0302 & 0.0236 & 0.0226 & 0.0224 & 0.0250 & \textbf{0.0222} & 0.0186 \\
Zipf $\alpha=5$ & 3.4539 & 2.5828 & 0.0182 & 0.0144 & 0.0138 & \textbf{0.0137} & 0.0162 & 0.0139 & 0.0112 \\
geometric & 2.2305 & 1.9730 & \textbf{1.5666} & 1.5987 & 1.5981 & 1.7038 & 1.8034 & 1.6402 & 1.5629 \\
Dirichlet-1 & \textbf{0.5590} & 0.5792 & 0.5603 & 0.5842 & 0.5601 & 0.6516 & 1.6805 & 0.5591 & 0.5588 \\
Dirichlet-$\frac{1}{2}$ & 0.9177 & \textbf{0.8951} & 0.8970 & 0.9268 & 0.8956 & 1.0083 & 1.8002 & 0.9044 & 0.8947 \\
\bottomrule
\end{tabular}
\end{table*}

\paragraph{Prediction and adaptation.}
At $n=1000$, the depth mixture is within $4.70\%$ of the lowest
non-oracle mean on every target. Its advantage over add-constant rules is
large on concentrated targets: on Zipf-$2$, add-one and KT incur $3.3625$
and $2.5263$ bits versus $0.1586$ for LSA. LSA is below Good--Turing on
Zipf-$1.5$--$5$ by $2.33$--$6.07\%$ in the primary sample set, with the
same direction in Table~\ref{tab:powers-bench}; Good--Turing is slightly
better on Zipf-$1$. Against the adaptive
Dirichlet mixture, it is lower on Zipf-$1$--$4$, but higher on step,
geometric, Dirichlet-$1/2$ and Zipf-$5$.

This comparison depends on sample size. Figure~\ref{fig:orlitsky} shows
material limitations on step and Zipf-$1$. At $n=20000$, the step losses are
$0.188722$ for LSA and $0.127061$ for the Dirichlet mixture, a $48.53\%$
excess relative to the latter, with paired difference
$0.061661\pm0.000039$ bits. On Zipf-$1$, LSA gives $0.153707$ versus
Good--Turing's $0.128592$, a $19.53\%$ excess, with paired difference
$0.025115\pm0.000103$ bits (one paired standard error in each case). Both
disadvantages have the same sign in all
ten trial blocks. Conversely, LSA is below Good--Turing on Zipf exponents
$2$ through $5$ at every tested sample size, with the same direction in
every block. Robust AD is below LSA on Zipf-$1$, $1.5$ and $5$ at every
tested size, and on geometric from $n=3000$ and Zipf-$2$ from $n=2000$;
LSA is lower on the other targets. Thus the near-best comparison at
$n=1000$ does not extend uniformly to the learning curves. These results illustrate both the useful
adaptation and the finite-sample limitations of this particular prior and
depth grid.

\begin{figure*}[!htbp]
\centering
\includegraphics[trim={0 36bp 0 0},clip,width=0.88\linewidth]{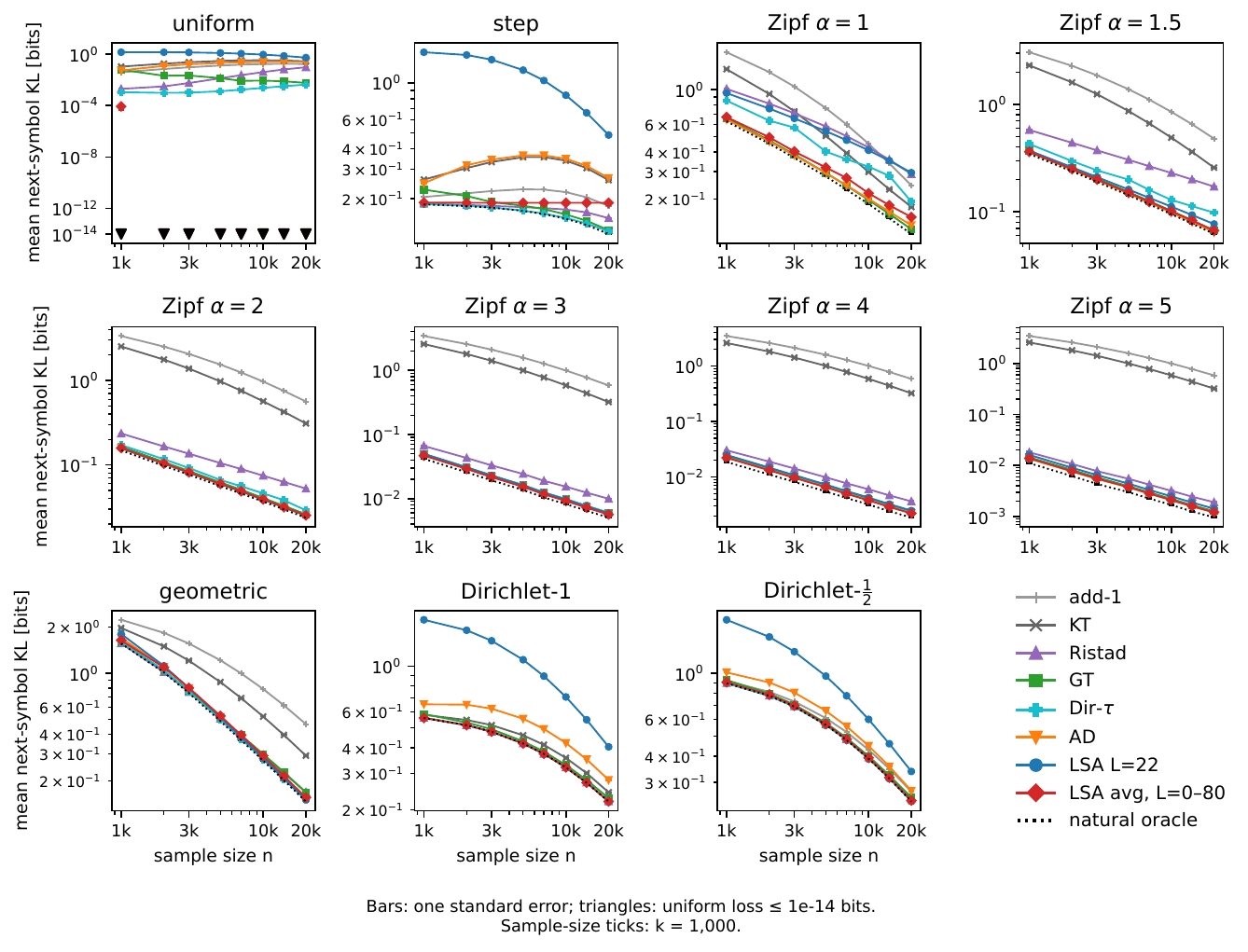}
\caption{Mean next-symbol KL loss in bits on eleven targets at $d=10^4$
and eight sample sizes from 1,000 to 20,000, with logarithmic axes ($1\mathrm{k}=1{,}000$ observations).
Each target--sample-size cell uses 1,000 shared trials across methods;
bars show one Monte Carlo standard error. The red diamonds denote the
LSA equal-prior depth mixture over $0,\ldots,80$; the legend identifies the
other methods. Black downward triangles in the uniform panel mark losses at or
below $10^{-14}$ bits.
Protocol: Appendix~\ref{app:experiments}.}
\label{fig:orlitsky}
\end{figure*}

The mean posteriors (Table~\ref{tab:posteriors}) help interpret these
results. Including $L=0$ gives $0.000088$ bits on uniform at $n=1000$;
Appendix~\ref{app:experiments} reports the ablation. On step, the
$n=20000$ mean posterior concentrates on the uniform component, consistent
with its $0.188722$-bit plateau. On Zipf-$1.5$, most mean posterior mass
moves from depths $12$--$19$ to $9$--$10$ as $n$ grows from $1000$ to
$20000$, consistent with the comparator scale $\cstar\ln(d/k)$ in
Theorem~\ref{thm:optimal}. That scale bounds cumulative regret rather
than selecting a finite-sample optimal depth. For exponents $3$--$5$, posterior mass spans many depths up to $80$. Performance with larger cutoffs remains unclear.

\paragraph{Coding a text corpus.}
The canonical King James Bible stream has $915{,}860$ tokens and
$13{,}550$ types, embedded in an alphabet of size $10^5$.
At the full prefix, LSA, Good--Turing and the Dirichlet mixture incur
$0.094978$, $0.096565$ and $0.101651$ bits per token of cumulative
redundancy (Table~\ref{tab:bible}). The LSA posterior mode decreases
from $21$ to $15$ across the five prefixes. LSA also has lower redundancy than Good--Turing
at $n=100$ and $500$ (Appendix~\ref{app:experiments}). We assume the dictionary and tokenization rules are known, exclude their transmission costs, and ignore linguistic context from preceding tokens. Appendix~\ref{app:memory} extends this experiment to order-one context models.

\begin{table*}[!htbp]
\centering
\caption{Cumulative coding redundancy in bits per token above the empirical
unigram entropy $\hat H_n$ of each canonical King James Bible prefix,
with $d=10^5$. The LSA mixture has equal prior weights on depths $0,\ldots,54$;
its posterior mode is shown. All methods use the same token stream and
condition on its dictionary and tokenization, whose transmission costs are
excluded. Bold minima and ties use three displayed decimal places. This is a
deterministic corpus comparison, not next-symbol KL loss. Estimators and
corpus preparation: Appendix~\ref{app:experiments}.}
\label{tab:bible}
\small\setlength{\tabcolsep}{5pt}
\begin{tabular}{rrrrrrrrrr}
\toprule
$n$ & distinct & $\hat H_n$ & Add-one & KT & Ristad & GT & Dir-$\tau$ & LSA & mode $L$ \\
\midrule
10000 & 1160 & 7.56 & 4.597 & 3.846 & 1.193 & 1.066 & 1.081 & \textbf{1.052} & 21 \\
30000 & 2119 & 7.87 & 3.225 & 2.512 & 0.709 & 0.606 & 0.621 & \textbf{0.598} & 21 \\
100000 & 3839 & 8.02 & 1.911 & 1.359 & 0.381 & 0.310 & 0.317 & \textbf{0.306} & 21 \\
300000 & 6922 & 8.18 & 1.026 & 0.679 & 0.224 & 0.170 & 0.176 & \textbf{0.168} & 19 \\
915860 & 13550 & 8.48 & 0.478 & 0.303 & 0.137 & 0.097 & 0.102 & \textbf{0.095} & 15 \\
\bottomrule
\end{tabular}

\end{table*}

\FloatBarrier
\section{Related work and discussion}
\label{sec:discussion}

\paragraph{Large-alphabet estimation.} Good--Turing~\citep{good1953population,gale1995good} and its analyses~\citep{orlitsky2003always,mcallester2000leave,acharya2013tight,orlitsky2015competitive,valiant2016instance,ohannessian2012rare,falahatgar2017power}, Pitman--Yor models~\citep{teh2006hierarchical}, and hierarchical Dirichlet priors with a learned or mixed concentration parameter~\citep{mackay1995hierarchical,friedman1998efficient,hutter2013sparse} address the same problem. Methods for assigning probability to unseen symbols, including prediction by partial matching (PPM) and absolute discounting~\citep{cleary1984data,witten1991zero}, are also relevant. Our evidence formula~\eqref{eq:gammatrick} is related to posterior analysis of normalized random measures~\citep{james2009posterior}. For a fixed alphabet, KT~\citep{krichevsky1981performance} achieves the optimal leading term in worst-case cumulative regret as the sample size grows~\citep{xie2000asymptotic}. However, it can be costly when only a small fraction of the alphabet is observed.

\paragraph{Discovery and coding redundancy.}
The pattern--dictionary decomposition separates repetition structure from symbol identities~\citep{orlitsky2004unknown}. \citet{falahatgar2015universal} establish order bounds for unordered power-law classes when $d>N$. Applying the Zipf count estimates in Appendix~\ref{app:optimal} to the adaptive sparse-predictor bound of \citet[Theorem~5]{hutter2013sparse} also yields the leading upper bound $k_N\log_2(d/k_N)$. That adaptive process generally depends on observation order. Our result establishes this sharp leading term for the architecture-motivated LSA prior, through an exchangeable predictor obtained from a fixed Bayesian depth mixture, for fixed $\alpha>1$ and $d=N^{\nu+o(1)}$, $\nu>1/\alpha$; this regime includes $d<N$. Both LSA and existing adaptive sparse predictors operate without knowing $\alpha$.

\begin{comment}
\paragraph{Scope.}
The optimality theorem concerns unit-power depth mixtures in its stated growing-alphabet Zipf regime. Finite-sample behavior is examined through the empirical approximation~\eqref{eq:law} and the single-weight heuristic in Appendix~\ref{app:experiments}; the Bible is a text case study. Powered layers preserve the evidence representation, without extending the optimality claim.
\end{comment}

\section{Conclusion}
The layered simplex architecture makes the connection from architecture to prior to prediction explicit. The architecture has a broad complexity spectrum. At any fixed positive KL tolerance, the depth mixture assigns targets uniform on sparse supports a prior complexity asymptotic to the cost of describing their support. As these targets vary, the prior probabilities of their KL neighborhoods span many orders of magnitude. Exact evidence and regret expressions make the predictive consequences accessible to analysis. Bayesian averaging adapts the weights over depths. In the stated growing-alphabet Zipf regime, the depth mixture attains the sharp leading-order regret governed by symbol discovery. The experiments show broad competitiveness with established predictors, alongside clear finite-sample limitations. Together, these results show how elementary multiplicative layering can induce a useful prior while retaining analytical tractability. This provides a concrete setting for studying the role of depth in learning architectures, a usable Bayesian predictor, and a starting point for investigating whether analogous connections help explain the predictive performance of modern layered learning architectures such as neural networks and transformers.

\section*{Acknowledgements}
We are grateful to Prof.\ Gregory Wornell. A question he asked one of us (MF) sparked the line of inquiry that led to this work. We also thank Ryan Jeong for providing
the asymptotic analysis of the layered prior \cite{jeong2026phase} and for discussions on the presented model.

\section*{AI Disclosure}
Claude (Anthropic) assisted implementation of numerical methods and experiment drivers; Claude and ChatGPT (OpenAI) assisted formulation and refinement of claims and proofs. AI tools also assisted experimental design and critique, independent numerical checks, interpretation of results, manuscript drafting and editing, literature identification, reference formatting, and figure preparation. ChatGPT assisted consistency and layout checks. The authors checked AI-assisted proofs by hand and are responsible for all claims, results, and AI-assisted content.

\bibliography{references_alt}

\clearpage
\appendix

\section{Prior-mass and regret identities}
\label{app:regret}

\begin{proofof}{the prior-mass bound~\eqref{eq:priormass}}
Fix $N\ge1$ and $\epsilon\ge0$ with $w_\epsilon=w(\Theta_0^\epsilon)>0$. Let $w_B$ be the prior restricted to $\Theta_0^\epsilon$ and normalized to have total mass one. For every $x^N$ with $p^N_{\theta_0}(x^N)>0$, $Q_N(x^N)\ge w_\epsilon\,p^N_{\theta_0}(x^N)\,\E_{\theta\sim w_B}[p^N_\theta(x^N)/p^N_{\theta_0}(x^N)]$. Apply $-\log_2$ to both sides. Since this function is convex, Jensen's inequality and expectation under $p^N_{\theta_0}$ give \[\KL(p^N_{\theta_0}\|Q_N)\le-\log_2w_\epsilon+N\,\E_{w_B}\bigl[\KL(p_{\theta_0}\|p_\theta)\bigr]\le-\log_2w_\epsilon+N\epsilon.\] Divide by $N$ and take the infimum over $\epsilon$ (zero mass gives an infinite bound).
\end{proofof}

\begin{proofof}{Theorem~\ref{thm:regret}}
For $X^N\sim p^N$ iid, $\E\bigl[\log_2 p^N(X^N)\bigr]=\sum_t\E\bigl[\log_2p(X_t)\bigr]=-NH(p)$. Hence
$\frac1N\KL(p^N\|Q_N)=\frac1N\E\bigl[\log_2p^N(X^N)\bigr]-\frac1N\E\bigl[\log_2Q_N(X^N)\bigr]=-H(p)-\frac1N\E\bigl[\log_2 q_{M}\bigr]$ with $M=m(X^N)\sim\Mult(N,p)$. Since $q_m$ depends on $m$ only through $\profile(m)$, grouping count vectors by profile gives~\eqref{eq:regret}. For the batch regret, $\hat q_m(i)=Q(X_{N+1}=i\mid X^N)=Q_{N+1}(x^N i)/Q_N(x^N)=q_{m+e_i}/q_m$, and $\KL(p\|\hat q_m)=-H(p)-\sum_ip_i\log_2\hat q_m(i)$. For fixed $m$, symbols with the same observed count have equal predictive probabilities. These probabilities may depend on the entire profile. Grouping over $r\ge0$ with $c_r>0$ gives~\eqref{eq:batch}; the unseen class is included when present. For $N=0$, $q_0=1$ and symmetry gives $\hat q_0(i)=1/d$. The identity \[\sum_iq_{m+e_i}=\E\Bigl[\prod_j\theta_j^{m_j}\sum_i\theta_i\Bigr]=q_m\] shows that the predictive probabilities sum to one. It also shows that summing over the last symbol in $Q_{N+1}$ gives $Q_N$. The chain rule $Q_N(x^N)=\prod_{n<N}\hat q_{m(x^n)}(x_{n+1})$ under the iid target then gives $\KL(p^N\|Q_N)=\sum_{n<N}\E\bigl[\KL(p\|\hat q_{M_n})\bigr]$, i.e.\ $NR_N=\sum_{n<N}\batchreg_n$.
\end{proofof}

\paragraph{The natural oracle.}
Fix a count vector $m$. For each nonempty count class $C_r=\{i:m_i=r\}$, let $c_r=|C_r|$ and $S_r=\sum_{i\in C_r}p_i$, including $r=0$ when symbols are unobserved. A natural estimator assigns $\hat q_m(i)=z_r$ within this class, with total estimated probability $T_r=c_rz_r$. The natural oracle assigns $q^*_m(i)=S_r/c_r$. Adding and subtracting $\log_2(S_r/c_r)$ within each class with $S_r>0$ gives
\begin{equation}
\KL(p\|\hat q_m)=\KL(p\|q^*_m)+\KL\bigl((S_r)_r\big\|(T_r)_r\bigr).
\end{equation}
Classes with $S_r=0$ contribute zero. If $T_r=0<S_r$, both sides are infinite. The second term is nonnegative, so the oracle minimizes KL loss among natural estimators. This term also gives the excess loss of any natural estimator over the oracle.

\begin{proofof}{Proposition~\ref{prop:laplace}}
For $\theta\sim\Dir(a_1,\dots,a_d)$, $\E\bigl[\prod_i\theta_i^{m_i}\bigr]=\frac{\Gamma(\sum_ia_i)}{\Gamma(\sum_i a_i+N)}\prod_i\frac{\Gamma(a_i+m_i)}{\Gamma(a_i)}$. With $a_i\equiv1$ this is $\Gamma(d)\prod_im_i!/\Gamma(d+N)$, and \[\frac{q_{m+e_i}}{q_m}=\frac{(m_i+1)!}{m_i!}\cdot\frac{\Gamma(d+N)}{\Gamma(d+N+1)}=\frac{m_i+1}{N+d}.\]
\end{proofof}

\begin{proofof}{Proposition~\ref{prop:avg}}
$Q_N(x^N)\ge\pi_LQ^{(L)}_N(x^N)$ pointwise; take $-\log_2$, then the expectation under $p^N$, and divide by $N$. The same argument holds for a mixture over infinitely many depths. Nonnegative weights that sum to one give a valid probability distribution. For the proposed weights, \[\sum_{L=1}^{\infty}\frac{1}{L(L+1)}=\sum_{L=1}^{\infty}\left(\frac1L-\frac1{L+1}\right)=1.\]
\end{proofof}
For a given sequence, the code-length bound for depth $L$ is nearly attained when that depth's posterior weight is close to one. Positive prior weights and likelihoods give positive posterior weights at every finite sample size. The bound controls cumulative regret. The next-symbol regret is the increment $\batchreg_N=(N+1)R_{N+1}-NR_N$, which we evaluate separately.

\section{Computing \texorpdfstring{$q_\lambda$}{q-of-lambda}
 and the predictive probabilities}
\label{app:qlambda}

This appendix gives the kernel representations, numerical procedure and error bounds used to compute the sequence probabilities $q_\lambda$ and the predictive probabilities. We use the profile notation of Section~\ref{sec:regret}; unless shown explicitly, all quantities refer to one fixed depth $L$.

\paragraph{Integral representation.} The derivation in~\eqref{eq:gammatrick} gives, for $N\ge1$,
\begin{equation}
\label{eq:q-integral}
q_\lambda=\frac1{\Gamma(N)}\int_0^\infty t^{N-1}\prod_{r=0}^N\phi_r(t)^{c_r}\,dt,
\end{equation}
where $\phi_r(t)=\E[Y^re^{-tY}]$ and $Y=\prod_{\ell=1}^LE_\ell$ is the unnormalized weight of one symbol. For $N=0$, set $q_0=1$. The same reduction applies to independent nonnegative weights whose sum is positive almost surely; identically distributed weights allow the factors to be grouped by count. This Laplace-transform method also appears in normalized random measures~\citep{james2009posterior}.

\paragraph{Layer recursion and the depth-one check.} The recursion~\eqref{eq:phi-recursion-main} relates kernels at neighboring depths and provides an independent check of the contour representation below. At $L=1$, the explicit kernels $\phi_r^{(1)}(t)=\Gamma(r+1)/(1+t)^{r+1}$ turn~\eqref{eq:q-integral} into a beta integral:
\[
q_\lambda^{(1)}=\frac{B(N,d)}{\Gamma(N)}\prod_{r=0}^N\Gamma(r+1)^{c_r}
=\frac{\Gamma(d)}{\Gamma(d+N)}\prod_{r=0}^N(r!)^{c_r}.
\]
This recovers the Dirichlet$(1,\dots,1)$ moment formula of Proposition~\ref{prop:laplace}.

\paragraph{Mellin--Barnes representation.} For $L\ge1$ and $t>0$, an alternative exact representation rests on two facts. First, the Mellin transform of $Y$ factorizes: $\E[Y^{z-1}]=\prod_{\ell}\E[E_\ell^{z-1}]=\Gamma(z)^L$ for $\Re z>0$. Second, the exponential has the inverse-Mellin (Cahen--Mellin) representation
\begin{equation}
e^{-u}=\frac{1}{2\pi i}\int_{\sigma-i\infty}^{\sigma+i\infty}\Gamma(z)\,u^{-z}\,dz,\qquad \sigma>0,
\end{equation}
which is Mellin inversion of $\int_0^\infty u^{z-1}e^{-u}du=\Gamma(z)$. Inserting it with $u=tY$ into $\phi_r(t)=\E[Y^re^{-tY}]$ and taking the expectation inside,
\begin{equation}
\label{eq:mb}
\phi_r^{(L)}(t)=\frac1{2\pi i}\int_{\sigma-i\infty}^{\sigma+i\infty}\Gamma(z)\,t^{-z}\,\E\bigl[Y^{r-z}\bigr]\,dz=\frac1{2\pi i}\int_{\sigma-i\infty}^{\sigma+i\infty}\Gamma(z)\,\Gamma(r+1-z)^L\,t^{-z}\,dz,
\end{equation}
since $\E[Y^{r-z}]=\Gamma(1+r-z)^L$ exists for $\Re(r-z)>-1$. For $0<\sigma<r+1$, \[t^{-\sigma}\E[Y^{r-\sigma}]\int_{\mathbb R}|\Gamma(\sigma+iy)|\,dy<\infty.\] This bound justifies exchanging expectation and integration. The valid strip $0<\sigma<r+1$ has its left edge from the Cahen--Mellin integral and its right edge from the existence of the moment. On a vertical line $|t^{-z}|=t^{-\sigma}$ is constant and, by Stirling's formula, $|\Gamma(\sigma+iy)|\sim\sqrt{2\pi}\,|y|^{\sigma-1/2}e^{-\pi|y|/2}$, so the integrand decays like $e^{-(L+1)\pi|y|/2}$ times a polynomial. The exponential decay becomes faster as $L$ increases. Moving the contour across a pole requires the residue: crossing $z=0$ to the left contributes $\Gamma(r+1)^L$, the small-$t$ limit of $\phi_r$; crossing the order-$L$ pole at $z=r+1$ to the right gives the fixed-parameter tail $\phi^{(L)}_r(t)=\frac{\Gamma(r+1)}{(L-1)!}\,t^{-(r+1)}(\ln t)^{L-1}[1+o(1)]$ as $t\to\infty$, whose leading constant depends on $r$ and $L$.

\paragraph{Offline kernel construction.}
The identities above are exact; their numerical evaluation uses precomputed
kernels and a one-dimensional outer integral. Depths zero and one use
$q_\lambda^{(0)}=d^{-N}$ and the Dirichlet closed form, respectively. For
$2\le L\le138$, we tabulate $h_r^{(L)}(u)=\ln\phi_r^{(L)}(e^u)$ once,
independently of $d$ and the observed profile. Each depth has 369 count
anchors: every integer $0\le r\le256$, followed by geometrically spaced
anchors and padding beyond the largest supported query count,
$1{,}045{,}889$. The $u$ grid has spacing $0.02$, upper endpoint $80$,
and a left endpoint at or below $-60-L\ln(r+1)$ for each row.

We construct the rows from~\eqref{eq:mb}, using a vertical contour inside
$0<\Re z<r+1$, near its real saddle. The mesh accounts for local curvature
and distance to the nearest pole; the extent is increased until the
endpoint magnitude has fallen by 40 nats relative to the central value.
A finite small-$t$ expansion is used where its remainder check meets the
construction tolerance $10^{-13}$. The completed binary64 tables are
immutable and reused through memory mapping. Evidence evaluation reads
these tables without computing a contour for each profile.

\paragraph{Interpolation and the small-$t$ limit.}
Each row is interpolated in $u$ by an eight-point, degree-seven barycentric
polynomial. The compiled C evaluator uses the actual binary64 grid
coordinates and compensated summation. Exact count anchors use their own
rows. Other counts use twelve anchors in $\ln(r+1)$, with a balanced
stencil across the dense-to-geometric join. For
$v=u+L\ln(r+1)\le2L$, when all shifted queries lie below the stored upper
endpoint, we align the rows at common $v$ and interpolate the residual
$h_r^{(L)}(u)-L\ln\Gamma(r+1)$ by a degree-eleven polynomial.
Elsewhere we interpolate $h_r^{(L)}(u)$ at fixed $u$.

Below a stored row's left endpoint, and for non-anchor queries with
$v\le-60$, we use the moment limit $L\ln\Gamma(r+1)$. The row endpoints
ensure $v\le-60$ on this branch. The inequality $1-e^{-x}\le x$ gives
\[
 0\le1-\frac{\phi_r^{(L)}(e^u)}{\Gamma(r+1)^L}
 \le e^u\frac{\Gamma(r+2)^L}{\Gamma(r+1)^L}=e^v.
\]
Thus its analytic log-kernel truncation error is at most
$-\ln(1-e^v)$, before floating-point evaluation. Queries outside the
supported count, depth or upper-grid domain are rejected.

\paragraph{Outer integration and predictive probabilities.}
For $N>0$, set $t=e^u$ and define
\[
 \Psi_\lambda(u)=Nu+\sum_{r=0}^N c_rh_r(u),
 \qquad q_\lambda=\frac1{\Gamma(N)}\int_{\mathbb R}e^{\Psi_\lambda(u)}\,du,
\]
where $Nu$ includes the Jacobian $dt=e^u\,du$. We scan this log-integrand, resolve its
significant peaks and integrate retained interior windows by the
trapezoidal rule, with a separate small-$t$ tail treatment. The initial
spacing is $0.02$; unresolved peaks trigger halving down to $0.0025$.
The initial right endpoint is $35$, extended to $60$ and then $80$ when
the boundary check requires it. Unresolved peaks or insufficient windows
at these limits raise a numerical error.
For $\lambda=\profile(m)$ and $r=m_i$, the identity $\Gamma(N+1)=N\Gamma(N)$ gives
\begin{equation}
\hat q_m(i)=\frac{q_{\lambda\oplus r}}{q_\lambda}
=\frac1N\cdot
\frac{\int e^u\,(\phi_{r+1}/\phi_r)(e^u)\,e^{\Psi_\lambda(u)}\,du}
{\int e^{\Psi_\lambda(u)}\,du}.
\end{equation}
The base profile and every distinct one-symbol augmentation share the
same integration grid and refinements: the augmented exponent adds
$u+h_{r+1}(u)-h_r(u)$ to the base exponent. For $N=0$, symmetry gives
$\hat q_0(i)=1/d$. Profiles in a batch share kernel preparation and
repeated count classes. All requested depths are evaluated; evidence
mixtures and posterior weights are formed in the log domain. The
integration dimension is independent of $d$, while the work depends on
the number of distinct counts, depths and grid points.

\paragraph{Error propagation.}
Suppose approximate log-kernels satisfy $|\tilde h_r-h_r|\le\epsilon_r$
on the whole integration domain, including a rigorously bounded tail.
With $\Delta_\lambda=\sum_{r=0}^N c_r\epsilon_r$, pointwise
exponentiation gives $e^{-\Delta_\lambda}Z_\lambda\le\tilde Z_\lambda\le e^{\Delta_\lambda}Z_\lambda$
for $Z_\lambda=\int_{\mathbb R}e^{\Psi_\lambda(u)}\,du$. If quadrature has relative error at most
$\eta<1$, then
\[
 |\ln\tilde q_\lambda-\ln q_\lambda|\le\Delta_\lambda-\ln(1-\eta).
\]
The log error in a predictive ratio is at most the sum of the errors in
its numerator and denominator. If every depth has log-evidence error at
most $\delta$, the mixture log-evidence error is at most $\delta$, and
each posterior log-weight has error at most $2\delta$. Under these assumptions, the zero-count log-kernel error bound is multiplied by the number of unobserved symbols, $c_0=d-s$.
Appendix~\ref{app:validation} reports finite-domain measurements of
kernel, evidence, prediction and refinement discrepancies separately.

\section{Computing the profile weights \texorpdfstring{$A_{\lambda}(p)$}{A-of-lambda(p)}}
\label{app:alambda}
Let $\lambda$ have positive-count multiplicities $c_r$, $r\ge1$, with $\sum_{r\ge1}rc_r=N$, $\mathcal R_\lambda=\{r\ge1:c_r>0\}$ and $s=\sum_{r\ge1}c_r$ ($A_\lambda=0$ if $s>d$). Each symbol is either unused, contributing $1$, or used $r$ times, contributing $z_rp_i^r/r!$, so
\[
A_\lambda(p)=\Gamma(N+1)\,[z^c]\prod_{i=1}^d\Bigl(1+\sum_{r\in\mathcal R_\lambda}z_r\frac{p_i^r}{\Gamma(r+1)}\Bigr),
\]
where $[z^c]$ denotes the coefficient of $\prod_{r\in\mathcal R_\lambda}z_r^{c_r}$. The unused class has size $c_0=d-s$ and is represented by the constant term. For the uniform target $u_i=1/d$, the formula becomes \[A_\lambda(u)=\frac{\Gamma(N+1)}{d^N}\frac{d!}{(d-s)!}\prod_{r\ge1}\frac{1}{\Gamma(r+1)^{c_r}c_r!}.\] By symmetry, the coefficient is a polynomial in the sums $P_j=\sum_i p_i^j$, with $1\le j\le N$. We evaluate this polynomial recursively after computing the required sums. The computational cost can depend on $d$. For large $N$, we evaluate the expectation in~\eqref{eq:regret} by drawing $M\sim\Mult(N,p)$ and averaging the computed $\log_2 q_M$ values. The Monte Carlo average is unbiased for exact log-probabilities; Appendix~\ref{app:qlambda} separately quantifies numerical error propagation. The identity $\sum_\lambda A_\lambda=1$ and the formula for a uniform target provide checks on the computation. The description length $\E\bigl[\log_2\binom{d}{K_N}\bigr]$ in Section~\ref{sec:scaling} is estimated from the same sampled profiles. 

\section{The role of the depth}
\label{app:fixedd}

This appendix shows how depth changes the prior over probability vectors with $d$ coordinates. We first study $L$ proportional to $\ln d$. We then study the limits with either $L$ or $d$ fixed. 

\paragraph{Large-deviation calculation.} For logarithmic depth \(L=\lfloor c\ln d\rfloor\), with \(c>0\), consider a single unnormalized weight \(Y=\prod_{\ell=1}^{L}E_\ell\), where the \(E_\ell\) are independent unit exponentials. Define
\[
I(x)=\sup_{\beta>-1}\{\beta x-\ln\Gamma(1+\beta)\},
\qquad
\rho_0(c)=cI(1/c).
\]
Since \(\ln Y\) is a sum of independent log-exponential variables, Cramér’s theorem gives
\[
\Pr\{Y\ge d\}=d^{-\rho_0(c)+o(1)}
\qquad(d\to\infty).
\]
Thus, \(\rho_0(c)\) quantifies the probability cost of a prescribed unnormalized weight reaching the scale \(d\).

\paragraph{Properties of $\rho_0(c)$.} Let $\Lambda(\beta)=\ln\Gamma(1+\beta)$. Write $\psi(x)=\frac{d}{dx}\ln\Gamma(x)$ and $\psi_1(x)=\psi'(x)$. Then $\Lambda'(\beta)=\psi(1+\beta)$ and $\Lambda''(\beta)=\psi_1(1+\beta)>0$, so $\Lambda$ is strictly convex for $\beta>-1$. For $x>-\gamma$ the maximizer $\beta(x)$ in $\rate(x)=\sup_\beta[\beta x-\Lambda(\beta)]$ solves $\psi(1+\beta)=x$, and $\rate'(x)=\beta(x)$, $\rate''(x)=1/\psi_1(1+\beta(x))$. With $x=1/c$, $\rho_0'(c)=\rate(x)-x\rate'(x)=-\Lambda(\beta(1/c))$ and $\rho_0''(c)=1/[c^3\psi_1(1+\beta(1/c))]$. Since $\Lambda<0$ on $(0,1)$, $\Lambda(1)=0$ and $\Lambda>0$ beyond, and $\beta(1/c)$ decreases in $c$, $\rho_0$ decreases until $\beta(1/c)=1$, i.e.\ $1/c=\psi(2)=1-\gamma$, and increases afterwards; at $\cstar=1/(1-\gamma)$ its value is $\cstar[(1-\gamma)-\Lambda(1)]=1$ and $\rho_0''(\cstar)=(1-\gamma)^3/\psi_1(2)=(1-\gamma)^3/(\pi^2/6-1)\approx0.1172$. Numerically $\rho_0(1)=1.315$, $\rho_0(1.5)=1.075$, $\rho_0(2)=1.009$. As $x\to\infty$, Stirling's formula gives $\rate(x)\sim e^x$. Hence $\rho_0(c)\sim c e^{1/c}$ as $c\to0$. The rate function $\rate$ is applied to the average log weight $(\ln Y)/L$. Its threshold $(\ln d)/L$ tends to $1/c$. The mean of this average is $-\gamma$. Related joint asymptotics for normalized weights are studied by \citet{jeong2026phase}. Appendix~\ref{app:optimal} gives the self-contained regret argument; Appendix~\ref{app:experiments} explores a heuristic connection between $\rho_0$ and finite-sample regret.

\paragraph{Fixed-parameter limits}

\begin{proposition}[fixed depth, growing alphabet]
For fixed $L$, as $d\to\infty$, $H(\theta)=\log_2d-L(1-\gamma)/\ln2+o_P(1)$ and $\max_i\theta_i\to0$ in probability.
\end{proposition}
\begin{proof}
For each $d$, use the first $d$ weights from one infinite sequence of independent copies of $Y$. Since $\E[Y]=1$ and $\E[Y|\ln Y|]<\infty$, the strong law gives $\frac1d\sum_{i\le d}Y_i\to1$ and $\frac1d\sum_{i\le d}Y_i\ln Y_i\to\E[Y\ln Y]=\frac{d}{ds}\Gamma(1+s)^L\big|_{s=1}=L\psi(2)=L(1-\gamma)$. Writing $H_{\mathrm{nats}}(\theta)=\ln S-\sum_iY_i\ln Y_i/S$ with $S=\sum_iY_i$ gives the entropy statement. For iid integrable nonnegative $Y_i$, $\max_{i\le d}Y_i/d\to0$ a.s.\ (Borel--Cantelli on $\sum_i\Prob(Y_i>\epsilon i)<\infty$), so $\max_i\theta_i\to0$. This proof holds for fixed $L$. The logarithmic-depth regime is studied separately through large deviations, as above and in~\citet{jeong2026phase}.
\end{proof}
\begin{proposition}
\label{prop:fixedd}
Fix $d\ge2$. As $L\to\infty$, $\max_i\theta_i\to1$ in probability with a uniformly distributed maximizing index, $q^{(L)}_m\to1/d$ if exactly one coordinate of $m$ is positive, and $q^{(L)}_m\to0$ if at least two are.
\end{proposition}
\begin{proof}
Let $V_i=\ln Y_i$, $\mu=-\gamma$, and $\sigma^2=\pi^2/6$. Let $i^*$ index the largest $V_i$, and let $D_L$ be the difference between the largest and second-largest values. By the multivariate central limit theorem, $((V_i-L\mu)/(\sigma\sqrt L))_{i\le d}$ converges to independent standard normal variables. Their two largest values have a strictly positive gap almost surely. Thus $D_L/(\sigma\sqrt L)$ converges in distribution to a strictly positive random variable, so $D_L\to\infty$ in probability. Therefore, \[1\ge\theta_{i^*}\ge[1+(d-1)e^{-D_L}]^{-1}\to1.\] Ties have probability zero and symmetry makes $i^*$ uniform. The law of $\theta$ converges weakly to the uniform law on the $d$ vertices; the bounded continuous monomial $\prod_i\theta_i^{m_i}$ equals $1$ at the named vertex when $m$ has one positive entry and $0$ at every vertex when $m$ has at least two positive entries.
\end{proof}
For fixed $d$ and $N\ge2$, cumulative regret therefore diverges as $L\to\infty$ whenever the target assigns positive probability to at least two symbols. A sequence containing both symbols has positive target probability, while its $Q_N^{(L)}$-probability tends to zero. The remaining terms in the regret sum are bounded below. Here $d$ is fixed. The regime $L\approx c\ln d$ instead lets $d$ and $L$ grow together. The eventual divergence at fixed $d$ complements the finite-range depth comparisons and supports mixing over depths, as in Section~\ref{sec:depth}.

\section{Support information and the Zipf lower bound}
\label{app:support}

\begin{proofof}{Theorem~\ref{thm:support}}
If the divergence is infinite there is nothing to prove; otherwise every support size seen with positive $P$-probability has positive $Q$-probability. Any two sets of $k$ symbols can be mapped to each other by relabeling the alphabet. Therefore, under a label-invariant $Q_N$, all such sets have the same conditional probability: $Q(S_N=s\mid K_N=k)=\binom dk^{-1}$. Let $\mathrm{Unif}_k$ be the uniform distribution over sets of $k$ symbols. Since the sequence determines $S_N$, the chain rule gives \[\KL(P_N\|Q_N)=\KL(P_{S_N}\|Q_{S_N})+\E_{S_N}\bigl[\KL(P_{X^N\mid S_N}\|Q_{X^N\mid S_N})\bigr].\] Since $S_N$ determines $K_N$, the same rule gives \[\KL(P_{S_N}\|Q_{S_N})=\KL(P_{K_N}\|Q_{K_N})+\sum_kP(K_N=k)\,\KL(P_{S_N\mid k}\|\mathrm{Unif}_k).\] Here, \[\KL(P_{S_N\mid k}\|\mathrm{Unif}_k)=\log_2\binom dk-H_P(S_N\mid K_N=k).\] This is~\eqref{eq:support}; dropping the two nonnegative divergences gives the inequality. Only label invariance is used.
\end{proofof}
\emph{Why the correction is needed.} For a uniform target $P(S_N\mid K_N)$ is itself uniform and $Q_N=p^N$ has zero regret, so the naming term is exactly canceled. Inside the layered family: at $d=3$, $N=5$, $p$ uniform and $L=1$, exact enumeration over count vectors gives $\KL(p^5\|Q^{(1)}_5)=0.45890149\ldots<\E\bigl[\log_2\binom{3}{K_5}\bigr]=0.60659059\ldots$ bits; the verification script evaluates both.

\begin{proofof}{Corollary~\ref{cor:zipf}}
We use natural logarithms, so entropy is measured in nats. We convert the final bound to bits. Write $k=k_N$.\par\emph{Entropy of the support.} Let $J_i=\mathbf 1\{M_{N,i}>0\}$, $\pi_i=\Prob(J_i=1)=1-(1-p_i)^N\le\min\{1,Np_i\}=\min\{1,(a_N/i)^\alpha\}$. Subadditivity gives $H(S_N\mid K_N)\le H(S_N)\le\sum_ih(\pi_i)$ with $h$ the binary entropy. For $i\le\lfloor a_N\rfloor$ use $h\le\ln2$; for $i>a_N$ use $h(x)\le x\ln(e/x)$ and the fact that $x\ln(e/x)$ is increasing on $[0,1]$ to get $h(\pi_i)\le(a_N/i)^\alpha[1+\alpha\ln(i/a_N)]$. The function $g(x)=x^{-\alpha}(1+\alpha\ln x)$ decreases on $[1,\infty)$, so $\sum_{i>\lfloor a_N\rfloor}g(i/a_N)\le1+a_N\int_1^\infty g=1+a_N\bigl[\frac1{\alpha-1}+\frac{\alpha}{(\alpha-1)^2}\bigr]$. Hence $H(S_N\mid K_N)\le C_\alpha a_N+1$.\par\emph{Fluctuations of $K_N$.} For $i\ne j$, $\mathrm{Cov}(J_i,J_j)=(1-p_i-p_j)^N-(1-p_i)^N(1-p_j)^N\le0$, so $\mathrm{Var}K_N\le\sum_i\mathrm{Var}J_i\le k$. Since $K_N\ge1$, convexity and $\ln x\le x-1$ give $0\le\E[K_N\ln(K_N/k)]\le\E[K_N(K_N/k-1)]=\mathrm{Var}(K_N)/k\le1$. With $\binom dK\ge(d/K)^K$, $\E\bigl[\ln\binom d{K_N}\bigr]\ge k\ln(d/k)-\E[K_N\ln(K_N/k)]\ge k\ln(d/k)-1$. Subtracting the entropy bound in~\eqref{eq:support}, dropping the nonnegative divergence of support sizes and converting to bits proves~\eqref{eq:zipflb}.

\emph{Expected number of observed symbols.} By Proposition~\ref{prop:discovery} below, $k_N\sim\Gamma(1-1/\alpha)a_N$ when $a_N\to\infty$ and $a_N/d\to0$ (the tail term is $a_N^\alpha d^{1-\alpha}=a_N(a_N/d)^{\alpha-1}=o(a_N)$); as $d\to\infty$ in this regime $H_{d,\alpha}\to\zeta(\alpha)$, so $k_N\sim\Gamma(1-1/\alpha)\zeta(\alpha)^{-1/\alpha}N^{1/\alpha}$, and $C_\alpha a_N=O(k_N)$ proves the second statement. If $d=N^\nu$ with $\nu>1/\alpha$, then $\ln(d/k_N)\sim(\nu-1/\alpha)\ln N$.
\end{proofof}

\begin{proposition}[finite-alphabet discovery]
\label{prop:discovery}
Let $d\ge2$, $N\ge1$, $p_i=i^{-\alpha}/H_{d,\alpha}$, $\alpha>1$, $a=N/H_{d,\alpha}$, and $\epsilon_P=\frac N2\sum_ip_i^2e^{-Np_i}/(1-p_i)$. Then
\[
\Bigl|k_N-\Gamma(1-\tfrac1\alpha)\,a^{1/\alpha}\Bigr|\le1+\epsilon_P+\frac{ad^{1-\alpha}}{\alpha-1},\qquad \epsilon_P\le\frac{1}{2e(1-\max_ip_i)} .
\]
\end{proposition}
\begin{proof}
Let $k^P_N=\sum_i(1-e^{-Np_i})$. This is the expected number of distinct symbols when the sample size has a Poisson distribution with mean $N$. Since $\ln(1-p)\le-p$, $k_N\ge k^P_N$. With $\delta(p)=-\ln(1-p)-p=\sum_{j\ge2}p^j/j\le p^2/(2(1-p))$ and $1-e^{-z}\le z$, $e^{-Np}-(1-p)^N=e^{-Np}(1-e^{-N\delta(p)})\le Np^2e^{-Np}/(2(1-p))$; summing gives $0\le k_N-k^P_N\le\epsilon_P$, and $Npe^{-Np}\le1/e$ gives the uniform bound on $\epsilon_P$. Define $f_a(x)=1-e^{-ax^{-\alpha}}$ for $x>0$, and set $f_a(0)=1$. This function decreases from one to zero, and $\int_0^\infty f_a=\frac{a^{1/\alpha}}{\alpha}\int_0^\infty(1-e^{-v})v^{-1-1/\alpha}dv=\Gamma(1-1/\alpha)a^{1/\alpha}$ (integrate by parts; the boundary terms vanish for $\alpha>1$). Monotone sum--integral comparison gives $0\le\int_0^\infty f_a-\sum_{i\ge1}f_a(i)\le1$, and the truncated tail satisfies $0\le\sum_{i>d}f_a(i)\le a\sum_{i>d}i^{-\alpha}\le ad^{1-\alpha}/(\alpha-1)$. Combining the three estimates proves the claim.
\end{proof}
For fixed $\alpha>1$ and $d\ge2$, $\max_ip_i\le(1+2^{-\alpha})^{-1}$, so $\epsilon_P$ is bounded independently of $N$ and $d$. The bound on the tail error, divided by $N^{1/\alpha}$, is of order $(N^{1/\alpha}/d)^{\alpha-1}$. For an infinite alphabet and a sample size with a Poisson distribution of mean $N$, the Euler--Maclaurin formula gives the expected number of distinct symbols as \[\Gamma(1-1/\alpha)(N/\zeta(\alpha))^{1/\alpha}-\frac12+o(1),\] as $N\to\infty$, for fixed $\alpha>1$. This $-\frac12$ correction concerns the infinite, Poissonized model. The bounds above quantify finite-alphabet and fixed-sample-size effects. Independently evaluating~$k_N$ at $(\alpha,d,N)=(1.5,10^3,10^4)$ gives $421.56$ against $655.05$ from the infinite approximation, a $55\%$ overestimate; at $d=10^5$ the overestimate is $3.7\%$.

\section{Complexity spectrum and first-order optimality}
\label{app:optimal}
We first establish an interval-mass estimate shared by Proposition~\ref{prop:complexity-spectrum} and Theorem~\ref{thm:optimal}, then prove the complexity spectrum and the regret bound. We use natural logarithms, except in the proof of Proposition~\ref{prop:complexity-spectrum}, where divergence and complexity are measured in bits. The prior over depths is fixed before the data are observed. To bound its loss using Proposition~\ref{prop:avg}, we may choose a comparison depth after seeing the data.

\paragraph{A uniform local limit theorem under tilting.}
Let $E\sim\Exp(1)$ and $Z=\ln E$. Multiplying the density of $Z$ by $e^{\beta z}$ and normalizing is called exponential tilting. For $\beta>-1$, this gives the density $f_\beta(z)=\exp[(1+\beta)z-e^z]/\Gamma(1+\beta)$, with mean $\mu_\beta=\Lambda'(\beta)$ and variance $\sigma_\beta^2=\Lambda''(\beta)>0$. Let $f_{L,\beta}$ be the density of the sum of $L$ independent variables with density $f_\beta$.
\begin{lemma}
\label{lem:lclt}
For every compact $J\subset(-1,\infty)$, $\sup_{\beta\in J,\,|u|\le1}\bigl|\sqrt L\,f_{L,\beta}(L\mu_\beta+u)-(2\pi\sigma_\beta^2)^{-1/2}\bigr|\to0$; in particular $f_{L,\beta}(L\mu_\beta+u)\ge b_J/\sqrt L$ for all large $L$, uniformly on $J\times[-1,1]$, for some $b_J>0$.
\end{lemma}
\begin{proof}
The characteristic function is $\chi_\beta(t)=\Gamma(1+\beta+it)/\Gamma(1+\beta)$, integrable, so Fourier inversion applies; center and substitute $v=t\sqrt L$. Uniformly on $J$, $e^{-it\mu_\beta}\chi_\beta(t)=\exp[-\sigma_\beta^2t^2/2+O_J(|t|^3)]$ near zero and $\sigma_\beta^2$ has a positive minimum on $J$, so for a fixed small $\delta$ the modulus is at most $e^{-bt^2}$ on $|t|\le\delta$ with $b>0$ independent of $\beta$, giving a Gaussian dominating function after rescaling and uniform convergence to the centered Gaussian density for bounded $u$. On $\delta\le|t|\le T$, continuity, compactness and nondegeneracy of the absolutely continuous tilted laws give $\sup_\beta|\chi_\beta(t)|\le q<1$; on $|t|\ge T$ Stirling gives $|\chi_\beta(t)|\le C_J|t|^{\beta+1/2}e^{-\pi|t|/2}\le e^{-b'|t|}$ for $T$ large. After the $L$th power both tail integrals are negligible even when multiplied by $\sqrt L$.
\end{proof}

\begin{lemma}[interval mass at logarithmic depth]
\label{lem:interval-mass}
Let $Y=\prod_{\ell=1}^{L}E_\ell$ with independent $E_\ell\sim\Exp(1)$. Put $\mu_1=\Lambda'(1)=1-\gamma$, $T>1$, $a=\ln T$, and let $L$ be the integer nearest to $a/\mu_1=\cstar a$, so $|a-L\mu_1|\le\mu_1/2$. For every fixed $0<B<\infty$, there are constants $C_B$ and $a_B$ such that, for all $a\ge a_B$ and integers $1\le m\le N$ satisfying $\ln(N+1)\le Ba$,
\begin{equation}
\label{eq:onesymbol}
-\ln\Prob\{mT\le Y\le(m+1)T\}\le a+\tfrac12\ln L+C_B[1+\ln(m+1)].
\end{equation}
In particular, for each fixed integer $m\ge1$, there is a constant $b_m>0$ such that $\Prob\{mT\le Y\le(m+1)T\}\ge b_m/(T\sqrt L)$ for all sufficiently large $T$.
\end{lemma}
\begin{proof}
Set $x=(a+\ln m)/L$ and choose $\beta$ with $\Lambda'(\beta)=x$; the bound on $\ln(N+1)/a$ confines all such $\beta$ to one compact subset of $(-1,\infty)$ for large $a$ (their lower limit is $1$). Let $V_L=\ln Y=\sum_{\ell=1}^L\ln E_\ell$ and $h=\ln(1+1/m)\le\ln2<1$. Exponential tilting and Lemma~\ref{lem:lclt} give
\[
\Prob\{V_L\in[Lx,Lx+h]\}=e^{-L\rate(x)}\int_0^he^{-\beta u}f_{L,\beta}(Lx+u)\,du\ge\frac{b_Bh}{\sqrt L}e^{-L\rate(x)-C_Bh},
\]
and $-\ln h\le\ln(m+1)$. Since $\rate(\mu_1)=\mu_1$ and $\rate'(\mu_1)=1$, and $\rate''$ is bounded on the compact range of $x$, Taylor's theorem gives $L\rate(x)\le a+\ln m+\frac{C_B}{L}[1+(\ln m)^2]\le a+C_B'[1+\ln(m+1)]$, using $\ln m=O_B(L)$. Combining proves~\eqref{eq:onesymbol}; the constants grow with $B$. The fixed-$m$ statement follows by taking $N=m$ and exponentiating the bound.
\end{proof}

\begin{proofof}{Proposition~\ref{prop:complexity-spectrum}}
Label invariance and Maclaurin's inequality give $\E_{W_d}\bigl[\prod_{i\in S}\theta_i\bigr]\le d^{-k}$: the average of these products over all $k$-element subsets is maximized at the uniform vector. The KL-neighborhood condition implies $\prod_{i\in S}\theta_i\ge2^{-k\epsilon}k^{-k}$. Markov's inequality therefore yields $C_\epsilon(u_S)\ge k\log_2(d/k)-k\epsilon$.

For the reverse bound, put $T=2d/k$, take $L$ nearest to $\cstar\ln T$, and set $m=\max\{1,\lceil2/(2^\epsilon-1)\rceil\}$. This depth lies within the mixture's cutoff. Require
\[
mT\le Y_i\le(m+1)T\quad(i\in S),\qquad \sum_{i\notin S}Y_i\le2d=kT.
\]
On this event, $\theta_i\ge m/[k(m+2)]$, so $\KL(u_S\|\theta)\le\log_2(1+2/m)\le\epsilon$. Since $\E[Y_i]=1$, the remaining-weight condition has probability at least $1/2$. The fixed-$m$ case of Lemma~\ref{lem:interval-mass} gives $\Prob\{mT\le Y\le(m+1)T\}\ge b_\epsilon/(T\sqrt L)$ for sufficiently large $d/k$, with $b_\epsilon>0$. Independence and the mixture weight thus give
\[
W_d\{\theta:\KL(u_S\|\theta)\le\epsilon\}\ge\frac{1}{2(\Lmax+1)}\left(\frac{b_\epsilon}{T\sqrt L}\right)^k.
\]
Taking logarithms,
\[
C_\epsilon(u_S)\le k\log_2(d/k)+\frac{k}{2}\log_2\ln(2d/k)+O_\epsilon(k+\ln\ln d).
\]
The remainder is $o(k\log_2(d/k))$: $d/k\to\infty$ and $k\ln(d/k)\ge\ln d$ eventually. Finally, $\ln\binom dk=k\ln(d/k)+O(k)$.
\end{proofof}

\paragraph{Comparison with depth one.}
At any fixed $\epsilon>0$ measured in bits, Proposition~\ref{prop:complexity-spectrum} gives $C_\epsilon(e_i)\sim\log_2d$ for a vertex $e_i$. For comparison, the depth-one prior gives $C^{(1)}_\epsilon(e_i)=-(d-1)\log_2(1-2^{-\epsilon})=\Theta_\epsilon(d)$, since $\theta_i\sim\operatorname{Beta}(1,d-1)$; layering makes this target much cheaper under the prior.

\begin{proposition}[profile-wise code length]
\label{prop:profile}
Let a sequence have length $N\ge1$, $k$ distinct symbols with counts $m_1,\dots,m_k$, $T=2d/k$, and $L$ nearest to $\cstar\ln T$. Write $\hat p_N$ for the empirical distribution: its observed probabilities are $m_i/N$, and all others are zero. Its entropy here is $H(\hat p_N)=-\sum_{i=1}^k(m_i/N)\ln(m_i/N)$. For fixed $B$, as $d/k\to\infty$, uniformly over profiles with $\ln(N+1)\le B\ln(2d/k)$,
\begin{equation}
\label{eq:profile}
-\ln q^{(L)}_m-NH(\hat p_N)\le k\ln\frac dk+\frac k2\ln\ln\frac{2d}{k}+C_B\Bigl[k+\sum_{i\le k}\ln(m_i+1)\Bigr]+\ln2 .
\end{equation}
For a mixture over depths containing this comparator add $\ln(1/\pi_L)$: $\ln(\Lmax+1)$ for the uniform prior in~\eqref{eq:avg}, $\ln[L(L+1)]$ for $\pi_L=1/[L(L+1)]$.
\end{proposition}
\begin{proof}
The unseen weights satisfy $\E\bigl[\sum_{i>k}Y_i\bigr]=d-k$ at every depth, so by Markov $\Prob\{\sum_{i>k}Y_i\le2d\}\ge\frac12$. Independently require $m_iT\le Y_i\le(m_i+1)T$ for each observed symbol. On the intersection the total unnormalized mass is at most $T\sum_i(m_i+1)+2d=T(N+2k)$, so $\theta_i\ge m_i/(N+2k)$ and
\[
\prod_{i\le k}\theta_i^{m_i}\ge\Bigl(\frac{N}{N+2k}\Bigr)^N\prod_i\Bigl(\frac{m_i}{N}\Bigr)^{m_i}\ge e^{-2k}e^{-NH(\hat p_N)} .
\]
The event has probability at least $\frac12\prod_i\Prob\{m_iT\le Y_i\le(m_i+1)T\}$ by independence. Integrating the likelihood over it and applying Lemma~\ref{lem:interval-mass} with $a=\ln(2d/k)$,
$-\ln q^{(L)}_m-NH(\hat p_N)\le2k+\ln2+k\ln(2d/k)+\frac k2\ln L+C_B\sum_i[1+\ln(m_i+1)]$; since $L$ is nearest to $\cstar\ln(2d/k)$, replacing $\ln L$ by $\ln\ln(2d/k)$ costs $O(k)$, and $k\ln2+2k$ is absorbed into $C_B$. The mixture probability is at least any component's probability multiplied by its fixed prior weight. We may choose the comparison component after seeing the sequence; the mixture's prior remains fixed.
\end{proof}
The comparison depth is the nearest integer to $\cstar\ln(2d/k)$, so it equals $\cstar\ln(d/k)+O(1)$. Here we must control the weights of all $k$ observed symbols together. The remainder $\frac k2\ln\ln(2d/k)$ is why the theorem is first-order.

\begin{proofof}{Theorem~\ref{thm:optimal}}
Let $\beta=1/\alpha$ and $\ln d/\ln N\to\nu\in(\beta,\infty)$; then $a_N\asymp N^\beta$, $a_N/d\to0$, $k_N\sim\Gamma(1-\beta)a_N$ by Appendix~\ref{app:support}, and $\ln(d/k_N)\sim(\nu-\beta)\ln N$.

\emph{Count-shape term.} By Jensen on each marginal count, $\E\bigl[\sum_i\ln(M_{N,i}+1)\bigr]\le\sum_i\ln(1+Np_i)\le a_N\int_0^\infty\ln(1+x^{-\alpha})dx=\frac{\pi a_N}{\sin(\pi/\alpha)}$, using that the summand is decreasing in $i$ and $Np_i=(a_N/i)^\alpha$; integration by parts gives $\alpha\int_0^\infty(1+x^\alpha)^{-1}dx=\pi/\sin(\pi/\alpha)$. Unseen symbols contribute zero, so this is exactly the sum in~\eqref{eq:profile}, and it is $O_\alpha(k_N)$.

\emph{Typical profiles.} Let $G_N=\{K_N\le2k_N\}$. When the sample size has a Poisson distribution with mean $N$, the indicators for whether each symbol appears are independent. Their sum, denoted $K_{\mathrm{Pois}(N)}$, has mean $k_N^P\sim k_N$, using the notation of Appendix~\ref{app:support}. A Chernoff bound gives $\Prob\{K_{\mathrm{Pois}(N)}>2k_N\}\le e^{-bk_N}$. Conditioning on the sample size being exactly $N$, which has probability $\sim(2\pi N)^{-1/2}$, then gives $\Prob(G_N^c)\le C\sqrt Ne^{-bk_N}$. On $G_N$, $d/K_N\to\infty$ and $\ln(N+1)/\ln(2d/K_N)\le\ln(N+1)/\ln(d/k_N)=O_{\alpha,\nu}(1)$, so Proposition~\ref{prop:profile} applies with a fixed $B$; $\Lmax=\lceil\cstar\ln(2d)\rceil+1$ contains every comparator depth. By concavity of $x\ln(d/x)$, $\E[K_N\ln(d/K_N)]\le k_N\ln(d/k_N)$, so the contribution of $G_N$ to the expected excess code length is at most $k_N\ln(d/k_N)+\frac{k_N}2\ln\ln(2d)+O_{\alpha,\nu}(k_N)+O(\ln\ln d)$, the last term being the mixture penalty (also $O(\ln\ln d)$ for the infinite prior).

\emph{Exceptional profiles.} The mixture contains depth one, whose add-one predictor assigns at least $1/(d+n)$ at time $n$, so $-\ln Q_N(x^N)\le N\ln(d+N)+\ln(\Lmax+1)$ always (with $\ln 2$ in place of $\ln(\Lmax+1)$ for the infinite prior); multiplied by $\Prob(G_N^c)$ this is $o(1)$ in the polynomial regime, and since empirical entropy is nonnegative the same bound controls the excess above $NH(\hat p_N)$ there.

\emph{Conclusion.} Concavity of entropy gives $\E\bigl[H(\hat p_N)\bigr]\le H(p)$, so \[\begin{aligned}\KL_{\mathrm{nats}}(p^N\|Q_N)&=\E[-\ln Q_N(X^N)]-NH(p)\\&\le\E[-\ln Q_N(X^N)-NH(\hat p_N)]\\&\le k_N\ln(d/k_N)(1+o(1)),\end{aligned}\]
because $\ln\ln d=O(\ln\ln N)=o(\ln(d/k_N))$. Corollary~\ref{cor:zipf} gives the matching lower bound for every label-invariant code, the depth mixture included. Converting to bits proves~\eqref{eq:optimal}, and $k_N\sim\Gamma(1-1/\alpha)\zeta(\alpha)^{-1/\alpha}N^{1/\alpha}$ with $\ln(d/k_N)=\ln(d/N^{1/\alpha})(1+o(1))$ gives the rate.
\end{proofof}
The theorem establishes first-order optimality for cumulative expected regret on Zipf targets with fixed $\alpha>1$ and fixed $\nu>1/\alpha$, in the joint limit $d=N^{\nu+o(1)}$. Corollary~\ref{cor:permutation-minimax} below gives its minimax interpretation over unknown label assignments. The finite-sample prediction comparisons in Section~\ref{sec:experiments} complement this asymptotic analysis.

\begin{remark}[Larger depth cutoffs]
\label{rem:larger-cutoffs}
The conclusion of Theorem~\ref{thm:optimal} holds for any uniform depth cutoff, fixed before observing the data, satisfying
\[
\Lmax\ge\lceil\cstar\ln(2d)\rceil+1,
\qquad \ln(\Lmax+1)=o\bigl(k_N\ln(d/k_N)\bigr).
\]
Indeed, every comparator depth remains available, and the proof above changes only the mixture penalty to $\ln(\Lmax+1)$. This is negligible on typical profiles; its contribution on exceptional profiles is also negligible under the stated condition and the exponential bound on $\Prob(G_N^c)$. In particular, $\Lmax\sim C\ln d$ with any fixed $C>\cstar$ satisfies these conditions.
\end{remark}

Returning to divergences measured in bits, the label-invariant bound has the following minimax interpretation.
\begin{corollary}[Unknown assignment of labels]
\label{cor:permutation-minimax}
Under the assumptions of Theorem~\ref{thm:optimal}, let $p_\sigma(i)=p_{d,\alpha}(\sigma^{-1}(i))$ for each permutation $\sigma$ in the symmetric group $S_d$. Then
\[
\inf_{Q_N}\sup_{\sigma\in S_d}\KL(p_\sigma^N\|Q_N)
=(1+o(1))\,k_N\log_2(d/k_N),
\]
where the infimum is over all probability distributions on $\Xd^N$. The depth mixtures in Theorem~\ref{thm:optimal} attain this leading term.
\end{corollary}
\begin{proof}
For any $Q_N$ with finite worst-case divergence, define its symmetrization by $\overline Q_N(x^N)=(d!)^{-1}\sum_{\sigma\in S_d}Q_N(\sigma x^N)$, where $\sigma$ acts on each symbol. Convexity of divergence in its second argument and relabeling give
\[
\KL(p_{d,\alpha}^N\|\overline Q_N)
\le\frac1{d!}\sum_{\sigma\in S_d}\KL(p_\sigma^N\|Q_N)
\le\sup_{\sigma\in S_d}\KL(p_\sigma^N\|Q_N).
\]
The distribution $\overline Q_N$ is label-invariant and positive, since each Zipf target has full support and the worst-case divergence is finite. Corollary~\ref{cor:zipf} therefore gives the lower bound $k_N\log_2(d/k_N)-O_\alpha(k_N)$. Infinite worst-case divergence is harmless for this bound. Conversely, the mixtures of Theorem~\ref{thm:optimal} are label-invariant, so their divergence is the same for every $p_\sigma$; that theorem gives the matching upper bound.
\end{proof}

\section{Powered layers: computation, representation, and experiments}
\label{app:powers}

Powers, fixed or random, extend the prior family without changing the dimension of its evidence integral. We give the computation for both constructions, prove a representation for integer powers, and compare a simple fixed-power mixture with the depth mixture. The main-text optimality theorem applies to unit-power LSA mixtures.

\subsection{Model and normalization}
\paragraph{Fixed powers.}
For fixed powers $\beta_\ell>0$, shared across coordinates within each layer, let
\begin{equation}
Y_i=\prod_{\ell=1}^L E_{\ell i}^{\beta_\ell},\qquad
\theta_i=\frac{Y_i}{\sum_{j=1}^dY_j},\qquad
E_{\ell i}\stackrel{\mathrm{iid}}{\sim}\Exp(1).
\label{eq:powered-model}
\end{equation}
The normalizer of each simplex layer cancels: raising $U_i^{(\ell)}=E_{\ell i}/\sum_jE_{\ell j}$ to the common power $\beta_\ell$ gives the same final prior as~\eqref{eq:powered-model}. When every layer uses the same power $w$, this is equivalent to raising each probability in a unit-power draw to $w$ and normalizing: $\theta_i(w)=\theta_i(1)^w/\sum_j\theta_j(1)^w$. Increasing $w$ sharpens that draw; decreasing it flattens the draw.

For $L\ge1$ and fixed positive powers, the moment formula holds when $\Re z>-1/\max_\ell\beta_\ell$:
\begin{equation}
\E[Y^z]=\prod_{\ell=1}^L\Gamma(1+\beta_\ell z),\qquad
\E[\ln Y]=-\gamma\sum_\ell\beta_\ell,\qquad
\operatorname{Var}(\ln Y)=\frac{\pi^2}{6}\sum_\ell\beta_\ell^2.
\label{eq:powered-moments}
\end{equation}
Here $\gamma\approx0.57721$ is the Euler--Mascheroni constant introduced in Section~\ref{sec:depth}. % [added: reminder of gamma]
The full exponent vector determines the distribution; its sum and sum of squares specify the first two log-weight moments. At $L=0$, or with every fixed exponent zero, the normalized vector is uniform.

% [added: random powers, model -- begin]
\paragraph{Random powers.}
The powers may also be random. Let $\nu$ be a probability distribution on $(0,\infty)$, chosen before the data are observed, and let
\begin{equation}
Y_i=\prod_{\ell=1}^L E_{\ell i}^{w_{\ell i}},\qquad
\theta_i=\frac{Y_i}{\sum_{j=1}^dY_j},\qquad
E_{\ell i}\stackrel{\mathrm{iid}}{\sim}\Exp(1),\qquad
w_{\ell i}\stackrel{\mathrm{iid}}{\sim}\nu,
\label{eq:random-powers-model}
\end{equation}
with the powers independent of the exponential variables. The powers are part of the prior and are integrated out; they are not fitted to the sample. When $\nu$ is concentrated on a single value $\beta$, the model is~\eqref{eq:powered-model} with $\beta_\ell=\beta$ in every layer; $\beta=1$ gives LSA. The pairs $(E_{\ell i},w_{\ell i})$ are iid across coordinates, so the unnormalized weights $Y_i$ remain iid and the prior remains label-invariant.

The random powers act on the exponential variables of~\eqref{eq:exprep}, not on the simplex layers of~\eqref{eq:lsa}. For a power shared within a layer the two constructions agree, as shown above. A power that varies with the coordinate prevents the layer normalizer from cancelling: $(U_i^{(\ell)})^{w_{\ell i}}=E_{\ell i}^{w_{\ell i}}\bigl(\sum_jE_{\ell j}\bigr)^{-w_{\ell i}}$, and the last factor depends on $i$. Raising the simplex layers to random powers would give a different prior, whose unnormalized weights are dependent.

For~\eqref{eq:random-powers-model} with $w\sim\nu$, when $\nu$ has a finite second moment,
\begin{equation}
\E[\ln Y]=-\gamma L\,\E[w],\qquad
\operatorname{Var}(\ln Y)=L\Bigl(\frac{\pi^2}{6}\,\E[w^2]+\gamma^2\operatorname{Var}(w)\Bigr),
\label{eq:random-powers-moments}
\end{equation}
which agrees with~\eqref{eq:powered-moments} for a constant power. Positive moments of $Y$ may be infinite: for $\nu=\Exp(1)$, $\E[E^{w}]=\int_0^\infty e^{-w}\Gamma(1+w)\,dw=\infty$. The prior is still well defined, because every $Y_i$ is positive and finite almost surely.
% [added: random powers, model -- end]

\subsection{Evidence and Bayesian averaging}
Let $Y^{(\ell)}$ be the product of the first $\ell$ powered exponential factors for one symbol. For $t>0$ and integer $r\ge0$, put $\phi_r^{(\ell)}(t)=\E[(Y^{(\ell)})^r e^{-tY^{(\ell)}}]$, with $Y^{(0)}=1$. The one-dimensional formula~\eqref{eq:q-integral} holds because the unnormalized coordinates are independent. A fixed powered layer gives
\begin{equation}
\phi_r^{(\ell)}(t)=\int_0^\infty e^{-x}x^{\beta_\ell r}
\phi_r^{(\ell-1)}(tx^{\beta_\ell})\,dx,
\qquad \phi_r^{(0)}(t)=e^{-t}.
\label{eq:powered-recursion}
\end{equation}
For $L\ge1$ and fixed positive powers, the Mellin--Barnes alternative is
\begin{equation}
\phi_r(t)=\frac1{2\pi i}\int_{\sigma-i\infty}^{\sigma+i\infty}
\Gamma(z)\prod_{\ell=1}^L\Gamma\bigl(1+\beta_\ell(r-z)\bigr)t^{-z}\,dz,
\quad 0<\sigma<r+\frac1{\max_\ell\beta_\ell}.
\label{eq:powered-mellin}
\end{equation}
% [added: random powers, kernel -- begin]
\paragraph{Random powers: the kernel.}
In the random-power model~\eqref{eq:random-powers-model} the unnormalized weights are again iid, so~\eqref{eq:q-integral} holds unchanged and only the kernels differ. Let $E\sim\Exp(1)$ and $w\sim\nu$ be the variables of layer $\ell$ for one symbol, so that $Y^{(\ell)}=E^{w}\,Y^{(\ell-1)}$. Conditioning on $E=x$ and on $w$ gives $\E\bigl[(Y^{(\ell)})^re^{-tY^{(\ell)}}\mid x,w\bigr]=x^{wr}\phi_r^{(\ell-1)}(tx^{w})$. The layer step is therefore~\eqref{eq:powered-recursion} averaged over the power:
\begin{equation}
\phi_r^{(\ell)}(t)=\int_{(0,\infty)}\int_0^\infty e^{-x}x^{wr}\,
\phi_r^{(\ell-1)}(tx^{w})\,dx\,\nu(dw),
\qquad \phi_r^{(0)}(t)=e^{-t}.
\label{eq:random-powers-recursion}
\end{equation}
Every layer uses the same $\nu$, so depth $L$ is $L$ applications of one step. The kernels are finite for every $t>0$ because $y^re^{-ty}$ is bounded on $y\ge0$; they do not require finite moments of $Y$.

On a logarithmic scale the layer step becomes a convolution with one fixed function. Put $t=e^u$ and define
\begin{equation}
\kappa_r^{(\ell)}(u)=e^{ru}\phi_r^{(\ell)}(e^u)=\E\bigl[F_r\bigl(u+\ln Y^{(\ell)}\bigr)\bigr],\qquad F_r(s)=\exp(rs-e^s).
\label{eq:random-powers-kappa}
\end{equation}
Each layer adds to $\ln Y$ an independent term $Z=w\ln E$. For $w\sim\nu$ this term has density
\begin{equation}
g_\nu(z)=\int_{(0,\infty)}\frac1w\exp\bigl(z/w-e^{z/w}\bigr)\,\nu(dw),\qquad z\in\mathbb R,
\label{eq:random-powers-density}
\end{equation}
and therefore
\begin{equation}
\kappa_r^{(\ell)}(u)=\int_{\mathbb R}g_\nu(z)\,\kappa_r^{(\ell-1)}(u+z)\,dz,\qquad
\kappa_r^{(0)}(u)=\exp(ru-e^u).
\label{eq:random-powers-convolution}
\end{equation}
The function $g_\nu$ depends on neither the count $r$ nor the layer, and the power distribution enters the computation only through it. Unit powers give $g(z)=\exp(z-e^z)$; a fixed power $\beta$ gives $g(z)=\beta^{-1}\exp(z/\beta-e^{z/\beta})$; random powers average the latter over $\beta\sim\nu$. Since $\sum_rrc_r=N$, the integrand of the outer integral in Appendix~\ref{app:qlambda} is $e^{\Psi_\lambda(u)}=\prod_{r=0}^N\kappa_r^{(L)}(u)^{c_r}$, and the predictive ratio there uses $e^u(\phi_{r+1}/\phi_r)(e^u)=\kappa_{r+1}^{(L)}(u)/\kappa_r^{(L)}(u)$. The outer integration is otherwise unchanged.

\paragraph{Random powers: numerical evaluation.}
The recursion~\eqref{eq:random-powers-convolution} can be evaluated on a uniform grid in $u$, the variable of the outer integral. The density $g_\nu$ is tabulated once from~\eqref{eq:random-powers-density} by one-dimensional quadrature in $w$. Each layer is then one discrete convolution per count, and one pass through the layers gives the kernels of every depth $1,\ldots,\Lmax$. Two features need care. First, suppose $\nu$ has a density that tends to a positive limit $\nu_0$ as $w\to0$, as a uniform distribution on $[0,a]$ does. Powers near zero then put mass of $Z$ near zero, and $g_\nu(z)\sim\nu_0\,e^{-1}\ln(1/|z|)$ as $z\to0$. The singularity is integrable, but point values of $g_\nu$ on the grid misstate the mass near the origin. A product-integration rule avoids this: it integrates the singular part of $g_\nu$ exactly against a piecewise-polynomial interpolant of $\kappa_r^{(\ell-1)}$. Without it the error compounds over the layers and, through the factor $\kappa_0^{c_0}$, over the unobserved symbols. Second, $\kappa_r^{(0)}$ is concentrated on a width of order $r^{-1/2}$ around $u=\ln r$. For large counts this is narrower than the grid spacing, and the first layer is then computed on a refined grid from the closed form of $\kappa_r^{(0)}$. The benchmark below uses fixed powers only; we leave random-power experiments to future work.
% [added: random powers, kernel -- end]

\paragraph{Numerical evaluation of the single-layer family for fixed powers.}
The benchmark below uses $Y=E^w$, with $E\sim\Exp(1)$ and integer powers
$w=0,\ldots,80$. Powers zero and one use the uniform and add-one closed
forms. For $w\ge2$, we use positive quadrature. Put $v=-\ln(t)/w$ and
$\psi_r(v)=t^r\E[E^{wr}e^{-tE^w}]$. In the variable $a=\ln E$, the
logarithm of its integrand is
$a+wr(a-v)-e^a-e^{w(a-v)}$. We center this strictly concave function at
its mode and apply Gauss--Legendre quadrature on subintervals resolving
the mode and, for $r=0$, the sharp cutoff. Refinement doubles the order
from 32 to at most 512 nodes. The weights are recomputed from the
Legendre derivative at the returned nodes, symmetrized and normalized;
tangent bounds control the inner tails. Since $\sum_r rc_r=N$, the
change of variables gives, for $N>0$,
\[
 q_m^{(w)}=\frac{w}{\Gamma(N)}\int_{\mathbb R}
 \psi_0(v)^{d-s}\prod_{r\ge1}\psi_r(v)^{c_r}\,dv.
\]
The base evidence and count-class predictions share a positive adaptive
vector quadrature. Default settings use an inner log-integral refinement
tolerance of $2\times10^{-12}$ and an outer relative tolerance of
$2\times10^{-8}$. Wider intervals and tighter quadrature settings provide
the refinement checks in Appendix~\ref{app:validation}; raw predictive
mass is checked separately before normalization. The sampling protocol
is given in Appendix~\ref{app:experiments}.

Any finite collection of models $\mathcal A$, each with evidence $Q_a(x^n)$ and a prior weight $\pi_a$ fixed before the sample, gives
\[
Q(x^n)=\sum_{a\in\mathcal A}\pi_aQ_a(x^n),\qquad
\widehat q(i\mid x^n)=\sum_a\frac{\pi_aQ_a(x^n)}{\sum_b\pi_bQ_b(x^n)}\widehat q_a(i\mid x^n).
\]
For equal weights the cumulative comparison penalty is $\log_2|\mathcal A|$ bits. This cumulative comparison applies to grids of powers or depths. Next-symbol KL loss is evaluated separately in the benchmark below.

\subsection{Integer powers as products of Dirichlet layers}
\begin{theorem}[integer-power representation]
\label{thm:integer-powers}
\mbox{}\par
{\emergencystretch=2em
Let $k\ge1$ be an integer. Let $U\sim\Dir(1,\dots,1)$, and let $V^{(1)},\dots,V^{(k)}$ be independent with $V^{(j)}\sim\Dir(j/k,\dots,j/k)$, all on $\simplex$. Then
\par}
\begin{equation}
\left(\frac{U_i^{\,k}}{\sum_{l=1}^dU_l^{\,k}}\right)_{i=1}^d
\;\stackrel{D}{=}\;
\left(\frac{\prod_{j=1}^kV^{(j)}_i}{\sum_{l=1}^d\prod_{j=1}^kV^{(j)}_l}\right)_{i=1}^d .
\label{eq:integer-powers}
\end{equation}
That is, one uniform simplex layer raised to the common power $k$ induces the same prior on $\simplex$ as the coordinate-wise product of $k$ independent symmetric Dirichlet layers with parameters $1/k,2/k,\dots,1$, renormalized as in~\eqref{eq:lsa}.
\end{theorem}
\begin{proof}
For $a>0$, let $\mathrm{Gamma}(a,1)$ have density $x^{a-1}e^{-x}/\Gamma(a)$ on $(0,\infty)$. Normalizing $d$ independent draws gives a $\Dir(a,\dots,a)$ vector. Write $U_i=E_i/\sum_lE_l$ with independent $E_i\sim\Exp(1)$, and $V_i^{(j)}=G_i^{(j)}/\sum_lG_l^{(j)}$ with independent $G_i^{(j)}\sim\mathrm{Gamma}(j/k,1)$ across $i$ and $j$. The layer normalizers cancel in~\eqref{eq:integer-powers}. The unnormalized vectors $(E_i^k)_i$ and $(\prod_jG_i^{(j)})_i$ both have independent, identically distributed coordinates. It therefore suffices to prove the scalar identity
\begin{equation}
E^k\;\stackrel{D}{=}\;k^k\prod_{j=1}^kG_{j/k},
\label{eq:integer-powers-scalar}
\end{equation}
where $E\sim\Exp(1)$ and the $G_{j/k}\sim\mathrm{Gamma}(j/k,1)$ are independent. The common factor $k^k$ cancels when the vector is normalized.

\emph{The scalar identity.} For a positive random variable $X$, the Mellin transform $s\mapsto\E[X^s]$ on the imaginary axis is the characteristic function of $\ln X$, so it determines the law of $X$. For $\Re s>-1/k$, $\E[(E^k)^s]=\Gamma(1+ks)$, and $\E[G_a^s]=\Gamma(a+s)/\Gamma(a)$ gives $\E\Bigl[\bigl(k^k\prod_jG_{j/k}\bigr)^s\Bigr]=(k^k)^s\prod_{j=1}^k\Gamma(s+j/k)/\Gamma(j/k)$. Gauss's multiplication formula $\Gamma(kz)=(2\pi)^{(1-k)/2}k^{kz-1/2}\prod_{j=0}^{k-1}\Gamma(z+j/k)$ at $z=s+1/k$ reads
\[
\Gamma(1+ks)=(2\pi)^{(1-k)/2}k^{ks+1/2}\prod_{j=1}^k\Gamma(s+j/k),
\]
and its value at $s=0$ gives $\prod_{j=1}^k\Gamma(j/k)=(2\pi)^{(k-1)/2}k^{-1/2}$. Dividing,
\[
\Gamma(1+ks)=(k^k)^s\prod_{j=1}^k\frac{\Gamma(s+j/k)}{\Gamma(j/k)},
\]
so the two Mellin transforms agree for $\Re s>-1/k$, in particular on the imaginary axis, which proves~\eqref{eq:integer-powers-scalar}.
\end{proof}
For $k>1$, the layers in this representation have different Dirichlet parameters: $1/k,2/k,\ldots,1$. A power changes which probability distributions the prior favors, even in a single-layer model.

% Replaces only the empirical subsection; powered-family derivations remain separate.
\subsection{A simple fixed powers benchmark}
We compare a single layer raised to an integer power $w=0,1,\ldots,80$
with the unit-power depth family $L=0,1,\ldots,80$. Each mixture gives equal
prior weight to its 81 components and updates weights by the sequence
evidence. Both include the uniform model, at $w=0$ or $L=0$, and have a
cumulative codelength at most $\log_2 81$ bits above their best member on
each sequence. This guarantee concerns cumulative loss, not the next-symbol
KL losses reported in Table~\ref{tab:powers-bench}. The benchmark uses the
common protocol of Appendix~\ref{app:experiments}, with an independent
sample set from Table~\ref{tab:bench}.

\paragraph{Results.}
Changing how the architecture induces its prior helps some targets and hurts
others. For power minus depth, the paired difference is
$-0.024940\pm0.000222$ bits on geometric and
$-0.000285\pm0.000040$ on Zipf-$5$, the latter a $2.08\%$ reduction
relative to the depth mixture. Ristad remains below both mixtures on
geometric. The power mixture is higher on Zipf-$1$ (by
$0.011787\pm0.000425$ bits), Zipf-$1.5$ and Dirichlet-$1/2$.
These favorable and unfavorable directions persist in every trial block.
Smaller positive differences also occur on Zipf-$2$ and Zipf-$3$;
Zipf-$4$ is unresolved at this sampling precision, with five blocks in
each direction. Each $\pm$ is one paired standard error. The two architectures
therefore produce distinct, target-dependent behavior without a uniform
advantage for either family. Robust AD is below both mixtures on
Zipf-$1$ and $1.5$.

\begin{table*}[!htbp]
\centering
\caption{Mean next-symbol KL loss in bits at $n=1000$, $d=10^4$, over
1,000 shared trials per target from a separate sample set from
Table~\ref{tab:bench}. Unit-power LSA uses fixed depth $22$ or an equal-prior
mixture over depths $0,\ldots,80$; the power mixture uses one layer with
equal prior weights on $w=0,\ldots,80$. Both mixtures contain the uniform
model. GT denotes the Good--Turing/empirical hybrid, Dir-$\tau$ the
Dirichlet concentration mixture and AD robust absolute discounting. Bold minima
and ties use four decimal places and exclude the target-informed oracle.
No trial is discarded. Marginal standard errors
are below $0.001566$ bits; the prose reports paired standard errors of
within-table differences. Protocol: Appendix~\ref{app:experiments}.}
\label{tab:powers-bench}
\footnotesize\setlength{\tabcolsep}{3pt}
\begin{tabular}{lrrrr|rr|rr|r|r}
\toprule
& \multicolumn{4}{c|}{classical} & \multicolumn{2}{c|}{adaptive controls} & \multicolumn{2}{c|}{LSA} & powers & \\
target & add-1 & KT & Ristad & GT & Dir-$\tau$ & AD & $L=22$ & avg $0$--$80$ & avg $0$--$80$ & oracle\\
\midrule
uniform & 0.0395 & 0.1087 & 0.0019 & 0.0627 & 0.0011 & 0.0554 & 1.4366 & \textbf{0.0001} & \textbf{0.0001} & 0.0000 \\
step & 0.2048 & 0.2609 & 0.1864 & 0.2256 & \textbf{0.1853} & 0.2485 & 1.5284 & 0.1896 & 0.1896 & 0.1846 \\
Zipf $\alpha=1$ & 1.7371 & 1.3535 & 1.0137 & 0.6583 & 0.8535 & \textbf{0.6454} & 0.9534 & 0.6693 & 0.6811 & 0.6282 \\
Zipf $\alpha=1.5$ & 3.0554 & 2.3044 & 0.5778 & 0.3705 & 0.4276 & \textbf{0.3594} & 0.3715 & 0.3623 & 0.3648 & 0.3460 \\
Zipf $\alpha=2$ & 3.3626 & 2.5264 & 0.2353 & 0.1650 & 0.1716 & 0.1584 & \textbf{0.1580} & 0.1583 & 0.1586 & 0.1483 \\
Zipf $\alpha=3$ & 3.4399 & 2.5757 & 0.0656 & 0.0493 & 0.0485 & 0.0473 & 0.0494 & \textbf{0.0465} & 0.0468 & 0.0413 \\
Zipf $\alpha=4$ & 3.4506 & 2.5813 & 0.0314 & 0.0244 & 0.0234 & 0.0230 & 0.0254 & \textbf{0.0229} & \textbf{0.0229} & 0.0194 \\
Zipf $\alpha=5$ & 3.4539 & 2.5828 & 0.0182 & 0.0143 & 0.0138 & 0.0135 & 0.0160 & 0.0137 & \textbf{0.0134} & 0.0110 \\
geometric & 2.2299 & 1.9722 & \textbf{1.5649} & 1.5974 & 1.5965 & 1.7031 & 1.8013 & 1.6388 & 1.6138 & 1.5612 \\
Dirichlet-1 & \textbf{0.5588} & 0.5790 & 0.5603 & 0.5840 & 0.5600 & 0.6519 & 1.6803 & 0.5590 & 0.5590 & 0.5586 \\
Dirichlet-$\frac{1}{2}$ & 0.9176 & \textbf{0.8948} & 0.8966 & 0.9266 & 0.8952 & 1.0075 & 1.7994 & 0.9038 & 0.9178 & 0.8943 \\
\bottomrule
\end{tabular}
\end{table*}

\FloatBarrier
\Needspace{8\baselineskip}
\section{Validation and reproducibility}
\label{app:validation}
Table~\ref{tab:validation} reports measured discrepancies on finite test
domains. Predictive KL changes compare $\KL(p\|\widehat q)$ for the same
target and normalized predictions; raw probability mass is checked before
normalization. Sampling uncertainty is separate (Appendix~\ref{app:experiments}).

\paragraph{Kernels.}
Nominal absolute targets are $10^{-11}$ nats for $r\le3$ and
$3\times10^{-9}$ otherwise. Precision qualification uses independent
45- and 60-digit references $R_{45},R_{60}$, their common nearest
binary64 value $q$, and the larger adjacent spacing $\delta$. It requires
$r>3$, $4\delta>3\times10^{-9}$, $|R_{45}-R_{60}|\le10^{-25}$,
$|q-R_{60}|\le\delta/2$ and $|\widehat h-q|\le4\delta$.
This allows $4.5\delta$ absolute error including reference rounding;
the four-spacing budget covers construction, storage, interpolation and arithmetic.

The 3,376 table-reader queries cover low counts at every depth $2$--$138$
and general counts through $1{,}045{,}889$ at depths $2,22,53,54,80,138$,
including anchor gaps, interpolation joins, shifted transitions and $u=80$.
There are 3,021 nominal and 355 precision-qualified passes. All 716 required
independent references converged; other queries use refined direct kernels.
Maximum errors are $7.6482\times10^{-7}$ nats overall,
$3.70\times10^{-13}$ nats for low counts, and three spacings from $q$
in qualified cases. Scalar and matrix results coincide on these queries.
The separate 263-case independent suite meets the nominal targets at
$L\in\{2,5,22,53,54,80,138\}$, $r\le1{,}045{,}889$ and selected
$u\in[-195.31,35]$.

\paragraph{Evidence and predictions.}
The 29 table-versus-direct profiles include synthetic samples, Bible
prefixes through 915,860 tokens, and sparse and concentrated cases at
every depth $0$--$138$. Direct contours replace tables at every requested
$L\ge2$ with identical outer settings; predictive cases include all
distinct one-symbol augmentations. All meet thresholds of $10^{-5}$
bits/token for codelength changes, $10^{-5}$ bits for predictive KL
changes and $10^{-7}$ for raw mass error. The 27 grid-refinement cases
halve the actual outer grid, even when default adaptation has already refined it.

\begingroup
\footnotesize\setlength{\tabcolsep}{3pt}
\begin{longtable}{@{}p{0.29\linewidth}p{0.32\linewidth}p{0.32\linewidth}@{}}
\caption{Maximum absolute discrepancies on the finite test domains described
here, rounded upwards. Slash-separated
values follow the listed order. Components are individual depths or powers;
mixtures use the full declared grid.}\label{tab:validation}\\
\toprule
Check & Quantity and unit & Maximum\\
\midrule
\endfirsthead
\toprule Check & Quantity and unit & Maximum\\\midrule
\endhead
\bottomrule\endfoot
Independent kernels: 263 cases
& Log-kernel error, all / $r\le3$ (nats)
& $2.16\times10^{-9}$ / $1.42\times10^{-13}$\\[2pt]
Meijer-$G$ versus contour: $L=1,2,3$, $r=0,\ldots,3$, $u=-2,0,2$
& Log-kernel discrepancy (nats) & $10^{-44}$\\[2pt]
Simplex versus reference: $d=2$, $L=2,3$
& Sequence-probability discrepancy & $1.03\times10^{-14}$\\[2pt]
Table versus direct: 29 profiles
& Component / mixture codelength change (bits/token)
& $8.47\times10^{-8}$ / $3.34\times10^{-9}$\\
& Augmented codelength change (bits/token) & $6.36\times10^{-8}$\\
& Component / mixture predictive KL change (bits)
& $5.39\times10^{-8}$ / $3.37\times10^{-9}$\\
& Raw probability-mass error & $9.21\times10^{-9}$\\[2pt]
Outer-grid halving: 27 profiles
& Component codelength change (bits/token) & $1.47\times10^{-14}$\\
& Component predictive KL change (bits) & $6.51\times10^{-12}$\\[2pt]
Power refinement: 11 profiles, $d=10^4$, $N=10^3$, $w=0,\ldots,80$
& Component log-evidence change (bits) & $3.68\times10^{-11}$\\
& Component / mixture predictive KL change (bits)
& $2.58\times10^{-12}$ / $2.53\times10^{-13}$\\[2pt]
Bible chain: first 2,000 transitions, $d=10^5$, depths $0,\ldots,54$
& Cumulative chain residual (total bits) & $4.55\times10^{-11}$\\
& Next-symbol probability residual & $1.20\times10^{-13}$\\
\end{longtable}
\endgroup

Simplex tests use counts $(2,0),(1,1),(2,1),(3,1)$. Six prior Monte Carlo
checks use 200,000 independent draws each: $L=1,2,3$ with $d=2$, counts $(3,1)$,
or $d=3$, counts $(2,1,0)$; the largest discrepancy is below 1.731 standard
errors. Layer recursion at $L=2$, $r=0,\ldots,3$, $u=-2,0,2$ agrees
with the contour at the saved 45-digit precision. The Bible check compares
sequential-mixture probabilities with separately evaluated evidence ratios,
and cumulative log losses with profile evidence after each transition;
predictions use only preceding tokens. These finite checks complement the
conditional error bounds in Appendix~\ref{app:qlambda}.

\paragraph{Runtime and reproducibility.}
Five workloads ran concurrently on one SCITAS Linux x86-64 node, each using
one numerical process and thread with the same table/native engine.
Intervals from worker start to completion include initialization,
evaluation, output writing and verification,
excluding queueing and table construction. For 20 trials at $d=10^4$ over
depths $0$--$80$, Zipf-$1.5$ took 1.87/13.04 minutes at $N=1000/20000$;
step at $N=20000$ took 0.88 minutes. A 20-profile spectrum cell
($\alpha=1.5$, $d=10^6$, $N=1000$, depths $0$--$138$) took 3.44 minutes.
The canonical Bible workload (all prefixes, short-prefix comparisons and
2,000 chain transitions, depths $0$--$54$) took 51.02 minutes. These are
observed workload times, not isolated cold-start or repeated warm-cache timings.
Records retain source/settings, software environment, native binary/table
identity, input checksums and seeds, together with common samples,
per-trial losses, posterior weights and numerical diagnostics.
The implementation, experiment configurations, and figure/table reproduction
scripts and compact result summaries are publicly available in release~v1.1.0\ifanonsubmission; the identifying repository
URL is omitted for anonymous review\else\ at
\url{https://github.com/urbanke/layered-simplex-alt2027}\fi.
Corpus downloads are checked against the recorded input checksums.

\section{Additional experimental details}
\label{app:experiments}
\suppressfloats[t]

\paragraph{Sampling protocol.}
Grids and sample counts are fixed in advance, with recorded seeds and no
outcome-dependent stopping. Primary, power and spectrum cells each use
1,000 trials; depth curves use 1,000 shared profiles per target, and
factorial cells use 5,000. Methods within a cell share identical counts.
Depths, or depth coefficients, share profiles within each comparison.
Primary and power benchmarks use independent target/sample sets. Numerical
methods and accuracy checks are in Appendices~\ref{app:qlambda}
and~\ref{app:validation}.

The next-symbol KL benchmarks (bits) use $d=10^4$, fixed depth $L=22$ and an
equal-prior depth mixture over $0,\ldots,80$. The primary sizes are
$n\in\{1000,\allowbreak2000,\allowbreak3000,\allowbreak5000,\allowbreak7000,\allowbreak10000,\allowbreak14000,\allowbreak20000\}$; the power benchmark
uses $n=1000$ and adds an equal-prior single-layer power mixture over
$w=0,\ldots,80$. Both mixtures update weights by the evidence.
The eleven targets are uniform; a step with half the symbols at $1/(2d)$
and half at $3/(2d)$; Zipf $\alpha\in\{1,1.5,2,3,4,5\}$;
geometric $p_i\propto0.998^i$; and symmetric Dirichlet draws with
per-coordinate parameters $1$ and $1/2$. A Dirichlet target is redrawn
per trial and held fixed across primary sizes; conditional on that target,
the multinomial samples at different sizes are independent, not nested.

The spectrum reports online regret $R_N$ in bits per symbol at
$\alpha=0,0.3,\ldots,3$, using all depths from zero through $L_{\max}$.
Its $(d,N,L_{\max})$ settings are $(10^3,316,69)$, $(10^4,1000,92)$
and $(10^6,1000,138)$. The factorial experiment uses
\[
\begin{aligned}
d&\in\{1000,3162,10000,31623,100000\},\\
N&\in\{100,316,1000,3162,10000\},\qquad\alpha\in\{1.5,2,3,4\}.
\end{aligned}
\]
These 100 cells use $L=\max\{1,\operatorname{round}(\cstar\ln d)\}$ and naming-cost
and regret estimates from the same profiles. The depth experiment fixes
$d=10^5$, $N=1000$, $\alpha=2,3,4$ and
$L=\max\{1,\operatorname{round}(c\ln d)\}$, with
\[
\begin{aligned}
c\in\{&0.25,0.5,0.75,1,1.25,1.5,1.75,2,\\
&2.37,2.75,3,3.5,4,4.5,5,5.5,6\}.
\end{aligned}
\]

\paragraph{Uncertainty and reporting.}
Synthetic results are sample means. Error bars and $\pm$ denote one
Monte Carlo standard error: sample standard deviation divided by the
square root of the trial count. We compute paired standard errors from differences between methods on the same trials. Comparisons between the primary and power tables are not paired. Standard errors are not confidence intervals. Ten consecutive
100-trial blocks per benchmark cell
provide descriptive stability checks, not confidence levels or substitutes
for pooled estimates. Numerical error is assessed separately.

Numerical predictions are normalized before scoring, retaining raw mass
for validation. A cell is undefined if any trial is undefined; otherwise
it is infinite if any trial is infinite. No nonfinite trial is discarded;
counts and reasons retain both types when they coexist. At $n=1000$, the
primary LSA-minus-AD difference is $0.023442\pm0.000339$ bits on Zipf-$1$,
positive in all ten blocks. On Zipf-$2$ it is $-0.000053\pm0.000050$ bits,
with six negative and four positive block means (one paired standard error
in each case). Paired
summaries are undefined if either method has any nonfinite outcome.
Benchmark minima and ties use four displayed decimals, excluding the
target-informed oracle; Bible ties use three.

\paragraph{Benchmark estimators.} Here $c_t$ is the number of symbols appearing $t$ times, and $s$ is the number of observed symbols. \par\noindent\textbf{Add-one.} Laplace's rule. \par\noindent\textbf{Add-half.} The KT rule. \par\noindent\textbf{Good--Turing/empirical hybrid.} A symbol seen $t$ times receives (before normalization) the empirical mass $t/n$ if $t>c_{t+1}$ and the Good--Turing mass $(c_{t+1}+1)(t+1)/(nc_t)$ otherwise, each unseen symbol receiving the $t=0$ value. \par\noindent\textbf{Ristad's law.} The predictor is $\hat q(i)=(m_i+1)(n+1-s)/(n^2+n+2s)$ for $m_i>0$ and $s(s+1)/((d-s)(n^2+n+2s))$ for $m_i=0$ when $s<d$ (add-one if $s=d$). \par\noindent\textbf{Absolute discounting.} For $0<s<d$, we use $\delta=\min\{\max(c_1,1)/s,0.9\}$ as in the experiments of \citet[Section~8]{falahatgar2017power}: $\hat q(i)=(m_i-\delta)/n$ for observed symbols and $\delta s/[n(d-s)]$ otherwise. At $n=0$ the predictor is uniform; at $s=d$ it is empirical. \par\noindent\textbf{Dirichlet concentration mixture.} The exact evidence is $q_\tau(m)=\frac{\Gamma(\tau)}{\Gamma(\tau+n)}\prod_{i:m_i>0}\frac{\Gamma(m_i+\tau/d)}{\Gamma(\tau/d)}$. The component predictor is $(m_i+\tau/d)/(n+\tau)$. We use a uniform prior on the $29$-point grid $\tau/d=2^j$, $j=-24,\dots,4$, and posterior weights $\propto q_\tau(m)$. This mixture receives the same cumulative guarantee as Proposition~\ref{prop:avg}. \par\noindent\textbf{Natural oracle.} The predictor is $\hat q(i)=S_{m_i}/c_{m_i}$. It knows the target and must give equal probability to symbols with equal counts. \par\noindent\textbf{LSA.} The fixed-depth and depth-averaged predictors. \par

\paragraph{Corpus preparation and controls.}
The canonical King James Bible stream uses the body of Project Gutenberg
\#10, removing verse numbers, stray asterisks and hyphens. Words and
punctuation are separated with case preserved; possessive markers are
separate tokens and trailing curly apostrophes are dropped. The stream
has 915,860 tokens and 13,550 types. With $d=10^5$ and equal-prior LSA
depths $0,\ldots,54$, we evaluate prefixes
$n=10000,30000,100000,300000,915860$ and short controls $100,500$.
Each prediction uses preceding tokens; LSA evidence gives the same
sequence probability from counts. Redundancy is codelength per token
minus prefix empirical unigram entropy, excluding dictionary transmission.

A second saved tokenization separates Unicode word runs and punctuation,
giving 917,868 tokens and 13,554 types. Across 20 comparisons for add-one,
KT, Ristad and Good--Turing, the largest absolute redundancy difference
is approximately $0.001738$ bits/token; full-stream endpoints have different lengths.
The Dirichlet control uses the same 29-component mixture. At canonical
$n=100,500$, LSA gives $4.161161,2.991974$ bits/token versus
Good--Turing's $4.176956,3.019708$.

\paragraph{Finite-sample scaling.}
The empirical approximation
\begin{equation}
\label{eq:law}
N R_N(d,L,\alpha)\approx\E\Bigl[\log_2\tbinom{d}{K_N}\Bigr]
+B(c,\alpha,N)\,k_N
\end{equation}
allows a signed offset $B$ in bits per discovered symbol. Writing
$x_d=\E[\log_2\binom d{K_N}]/k_N$ and $y_d=NR_N/k_N$, we report the
free ordinary least-squares slope and, separately, the slope-one offset
$B=\operatorname{mean}_d(y_d-x_d)$, not the free-fit intercept.
Data exponents are ordinary least-squares fits on log--log axes.
The free alphabet slopes span $0.9419$--$1.1210$ for $\alpha\ge2$.
At $\alpha=1.5$, $N=10000$, the slope is $0.5798$ and $k_N=421.56$
in the smallest alphabet ($d=1000$), illustrating finite-alphabet effects.

\begin{table}[!htbp]
\centering
\caption{Finite-range scaling, with 5,000 profiles per cell and
$L=\max\{1,\operatorname{round}(\cstar\ln d)\}$. Left: free-fit slope
of $NR_N/k_N$ versus $\E[\log_2\binom d{K_N}]/k_N$ over five alphabets
$d=10^3$--$10^5$, and the separately fitted slope-one offset $B$
(bits/discovered symbol). Right: fitted exponents $s$ in $R_N\propto N^{-s}$
over five sizes $N=100$--$10000$, compared with $1-1/\alpha$.
Here $R_N$ is online regret in bits/symbol and $k_N=\E[K_N]$ is the exact
finite-alphabet discovery mean; fitted exponents need not equal asymptotic powers.}
\label{tab:alphabet}\label{tab:data}
\footnotesize\setlength{\tabcolsep}{3.5pt}
\begin{tabular}{c ccc ccc}
\toprule
& \multicolumn{3}{c}{free-fit slope} & \multicolumn{3}{c}{slope-one offset $B$} \\
$\alpha$ & $N=10^2$ & $10^3$ & $10^4$ & $N=10^2$ & $10^3$ & $10^4$ \\
\midrule
1.5 & 0.87 & 0.76 & 0.58 & -2.73 & -2.17 & -1.34 \\
2 & 1.01 & 1.00 & 0.94 & -0.83 & -0.63 & -0.47 \\
3 & 1.05 & 1.05 & 1.06 & 1.34 & 1.81 & 1.96 \\
4 & 1.05 & 1.07 & 1.12 & 2.72 & 3.46 & 3.87 \\
\bottomrule
\end{tabular}
\hfill
\begin{tabular}{c c cc}
\toprule
& asymptotic & \multicolumn{2}{c}{measured exponent} \\
$\alpha$ & $1-1/\alpha$ & $d=10^3$ & $d=10^5$ \\
\midrule
1.5 & 0.333 & 0.520 & 0.437 \\
2 & 0.500 & 0.607 & 0.546 \\
3 & 0.667 & 0.689 & 0.670 \\
4 & 0.750 & 0.736 & 0.728 \\
\bottomrule
\end{tabular}
\end{table}

For $k=o(d)$, Stirling's approximation gives
$\log_2\binom dk=k\log_2(ed/k)+O(\log k+k^2/d)$.
Concavity makes $\log_2\binom d{k_N}$ an upper bound on
$\E\bigl[\log_2\binom d{K_N}\bigr]$. Thus, with $P=NR_N/k_N$,
$P\approx\log_2(ed/k_N)+B_N$ and, writing
$\eta_N=d\ln k_N/d\ln N$, the local decay exponent at fixed $d$ is
\[
-\frac{d\ln R_N}{d\ln N}\approx1-\eta_N+
\frac{\eta_N\log_2 e-dB_N/d\ln N}{P}.
\]
For $\eta_N=1/\alpha$ and constant $B$, this becomes
$(1-1/\alpha)+\log_2 e/(\alpha P)$. Finite-range slopes average these
local effects; the discovery correction and logarithmic cost explain why
the measured exponents are not expected to equal the asymptotic power.

\paragraph{Depth and posterior adaptation.}
Relative to the tested reference $c=2.37$, near
$\cstar=1/(1-\gamma)$, the best tested coefficients $4,6,6$ for
$\alpha=2,3,4$ reduce mean cost by $1.15\%$, $5.53\%$, $8.67\%$.
The latter two lie at the upper grid boundary, not established global optima.
Figure~\ref{fig:depth} compares the data with the heuristic
\[
\rho(c)\frac{\E[\log_2\binom d{K_N}]}{k_N}+\alpha\log_2(e)-2,
\qquad
\rho(c)=\begin{cases}\rho_0(c),&c\le\cstar,\\1,&c>\cstar\end{cases},
\]
using the single-weight factor $\rho_0$ from Appendix~\ref{app:fixedd}
and a fixed, unfitted offset. It captures broad flattening but overestimates
shallow-depth costs; it is neither an error bound nor an optimal-depth rule.

At $n=1000$, adding depth zero to the mixture over $L=1,\ldots,80$ reduces
mean next-symbol KL loss from $0.039530$ to $0.000088$ bits on uniform
and from $0.204754$ to $0.189580$ on step, using the same primary samples.

Mean posteriors in Table~\ref{tab:posteriors} show adaptation within the
fixed grid. For $\alpha=3$--$5$ their broad intervals reach depth $80$,
which does not validate a larger cutoff. A rounded weight of $1.00$ need
not be exactly one.

\begin{figure*}[!htbp]
\centering
\includegraphics[width=\textwidth]{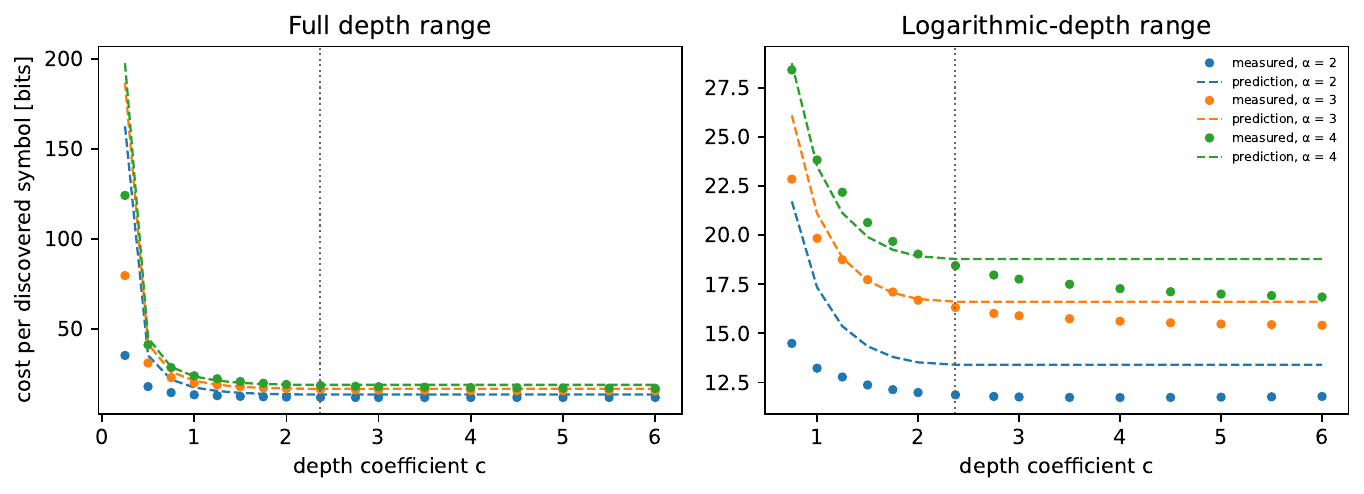}
\caption{Online regret per discovered symbol $NR_N/k_N$ in bits versus
$c$, at $d=10^5$, $N=1000$, $\alpha=2,3,4$. Each target uses the same
1,000 profiles across 17 coefficients, with
$L=\max\{1,\operatorname{round}(c\ln d)\}$.
The right panel enlarges the logarithmic-depth range. Dashed curves are
the single-weight heuristic with offset $\alpha\log_2(e)-2$ and
$\rho(c)=1$ above $\cstar$ (vertical reference); endpoint minima do not
establish global optimality.}
\label{fig:depth}
\end{figure*}

\begin{table}[!htbp]
\centering\footnotesize\setlength{\tabcolsep}{2pt}
\caption{Mean depth posteriors at $d=10^4$, $n=1000,20000$, from 1,000
primary trials per target and an equal prior on $L=0,\ldots,80$.
Broad summaries give the shortest contiguous interval containing at least
$95\%$ of mean posterior mass (ties favor greater enclosed mass), with
actual enclosed mass shown. Intervals of at least five depths are reported;
otherwise weights at least $0.05$ are listed to two decimals.
These summarize mean posterior mass, not confidence intervals or
trial-to-trial variation.}
\label{tab:posteriors}
\begin{tabular}{lll}
\toprule
target & $n=1000$ & $n=20000$ \\
\midrule
uniform & $L=0$: 1.00 & $L=0$: 1.00 \\
step & $L=0$: 0.85, $L=1$: 0.15 & $L=0$: 1.00 \\
Zipf $\alpha=1$ & $L=5$: 0.70, $L=6$: 0.30 & $L=3$: 1.00 \\
Zipf $\alpha=1.5$ & broad: 96.0\% on $L\in[12,19]$ & $L=9$: 0.66, $L=10$: 0.33 \\
Zipf $\alpha=2$ & broad: 95.4\% on $L\in[20,54]$ & broad: 95.5\% on $L\in[19,29]$ \\
Zipf $\alpha=3$ & broad: 95.6\% on $L\in[39,80]$ & broad: 95.2\% on $L\in[46,80]$ \\
Zipf $\alpha=4$ & broad: 95.0\% on $L\in[43,80]$ & broad: 95.7\% on $L\in[52,80]$ \\
Zipf $\alpha=5$ & broad: 95.2\% on $L\in[43,80]$ & broad: 95.8\% on $L\in[52,80]$ \\
geometric & $L=4$: 0.81, $L=5$: 0.19 & $L=9$: 0.07, $L=10$: 0.92 \\
Dirichlet-1 & $L=1$: 1.00 & $L=1$: 1.00 \\
Dirichlet-$\frac{1}{2}$ & $L=1$: 0.25, $L=2$: 0.75 & $L=2$: 1.00 \\
\bottomrule
\end{tabular}
\end{table}

\clearpage
\section{Adding memory}
\label{app:memory}
This appendix extends the memoryless Bible experiment of Section~\ref{sec:experiments} to predictors that use the preceding token as context.

In this extension, the context structure determines how the data are
partitioned, while the prior depth controls the distribution modeled
within each context. We use an order-one model with a separate state for
each of the $M$ most frequent preceding tokens and a single backoff state
for all other histories. Each state has its own depth-averaged predictor,
allowing the data to favor different depths in different contexts. We
also average over the number of contexts, adding another level of
hierarchical averaging \citep{FederMerhavHierarchical}. The resulting
structure corresponds to a context tree with a single split level
\citep{willems1995context,willems1998context}, with a depth-averaged
predictor at each leaf.

\paragraph{Experimental setup.}
The corpus is the King James Bible stream of $N=915{,}860$ tokens and
$13{,}550$ types over a fixed vocabulary of $d=10^5$. For each fixed
partition, all methods use the same token stream and context assignment:
contexts are selected by full-stream token counts, the first token is
routed to backoff, and all $N$ tokens are coded. Rates are total
codelengths in bits per token, conditional on the shared dictionary and
tokenization, whose transmission costs are excluded. The LSA
priors average uniformly over positive depths $1,\ldots,L_{\max}$;
unlike the main experiment, they do not include $L=0$.

For a context set $S_M$, let $x^{(s)}$ be the subsequence routed to state
$s$, with $\bot$ denoting backoff. The mixture factorizes over states:
\begin{equation}
\label{eq:memory-factorization}
Q_{S_M}(x^N)=\prod_{s\in S_M\cup\{\bot\}}
\left[\frac{1}{L_{\max}}\sum_{L=1}^{L_{\max}}Q^{(L)}(x^{(s)})\right].
\end{equation}
Empty subsequences contribute one. The total codelength is a sum of
per-state codelengths, each evaluated from successor counts by the
machinery of Section~\ref{sec:regret}. For a partition known to the
decoder, this defines a sequential probability assignment. The complete
sequence model is generally not exchangeable, and the main text's iid
regret guarantees do not automatically extend to it.

\begin{table}[!htbp]
\centering
\small
\begin{tabular}{lcc}
\toprule
 & $M=256$, $1\le L\le40$ & $M=512$, $1\le L\le60$ \\
\midrule
Coverage of transitions & $76.6\%$ & $83.7\%$ \\
Empirical $H(\text{next}\mid\text{state})$ & $6.315$ & $6.098$ \\
\midrule
LSA, per-state depth average & $\mathbf{7.132}$ & $\mathbf{7.100}$ \\
Good--Turing (per state) & $7.143$ & $7.113$ \\
Ristad (per state) & $7.254$ & $7.227$ \\
KT (per state) & $9.781$ & $10.172$ \\
Add-one (per state) & $10.445$ & $10.868$ \\
\bottomrule
\end{tabular}
\caption{Order-one models of the Bible (bits per token; unigram
reference: empirical entropy $8.477$,
LSA code $8.572$). Every method has the same type of memory: $M$ context
states plus backoff, with separate successor counts in each state.
Rates exclude the common
fixed-partition context-set header.}
\label{tab:memory-fixed}
\end{table}

\paragraph{Memory pays.}
The LSA codelength decreases from $8.572$ bits per token without memory
to $7.100$ with $512$ context states. Among the predictors in
Table~\ref{tab:memory-fixed}, LSA achieves the lowest codelength at both
partition sizes, followed by per-state Good--Turing.

\paragraph{The reversal for flat priors.}
The comparison also exhibits a reversal: order-one KT codes at
$9.78$ bits per token, \emph{worse than the memoryless layered model},
because every state is another $10^5$-symbol alphabet on which an
add-constant prior spreads its mass. Refining the partition makes KT and
add-one worse while the layered prior improves. In this experiment, the
choice of conditional prior matters substantially alongside the memory
it is attached to.

\paragraph{Averaging over the number of contexts.}
The number of contexts can itself be averaged over using a uniform
mixture of nested partitions:
\begin{equation}
\label{eq:memory-family}
Q_{\mathcal G}(x^N)=\frac{1}{|\mathcal G|}
\sum_{M\in\mathcal G}Q_{S_M}(x^N).
\end{equation}
We use the grids $\{0,16,32,64,128,192,256\}$ for depth cutoff 40 and
$\{0,64,128,256,384,512\}$ for cutoff 60. Conditional on the family
description, the mixture is within $\log_2|\mathcal G|$ bits of the best
member. Here $M=0$ means one memoryless state, not a depth-zero prior.

\paragraph{Describing the context set.}
The partition was defined by the $M$ most frequent tokens of the full
corpus, so a decodable two-part code must also specify which tokens these
are. The overhead is small. A two-part code naming the $512$ split
contexts out of the $d=10^5$ vocabulary costs
$\log_2\binom{10^5}{512}\approx4627$ bits, about $0.005$ bits per token,
raising the fixed-partition LSA rate from $7.100$ to
approximately $7.105$. The same header is owed by every method using
that fixed partition, so its inclusion does not change the ordering in
Table~\ref{tab:memory-fixed}.

For a nested family, the decoder additionally needs the allocation of
contexts to its successive groups. If the largest set has size $M_*$ and
the group sizes are $b_1,\ldots,b_J$, the shared header has ideal length
\begin{equation}
\label{eq:memory-header}
\log_2\binom d{M_*}+\log_2\frac{M_*!}{\prod_{j=1}^J b_j!}.
\end{equation}
The subset-only adjustment above therefore applies to a fixed partition,
not to the entire nested family. Grids and depth cutoffs are shared
protocol settings; integer rounding of the two header lengths costs
fewer than two additional bits.

\paragraph{Learning the context set online.}
The header can also be avoided entirely. A context gets its own state
once its count so far reaches a threshold. This code is fully sequential,
and its probability can still be evaluated from per-state count profiles,
because which state a token goes to depends only on the past. Earlier
successors remain in backoff and are not transferred retrospectively.

For the online model, promotion uses the number of earlier occurrences
as a context, excluding the current occurrence. Conditional on the first
token, the code assigns probabilities to the remaining $N-1$ transitions.
One can use the depth grid $\{1,5,10,15,20,25,30,35,40,45,50,60\}$
and thresholds $\{32,64,128,256,512\}$. A uniform mixture over the five
threshold models costs at most $\log_2 5$ bits more than its best member.
This construction avoids transmitting a context set selected from the
full corpus; Table~\ref{tab:memory-fixed} reports only the fixed-partition
models.

\end{document}